\documentclass[envcountsame,reqno]{llncs}

\usepackage{amssymb,gastex,mathtools,times,hyperref}
\usepackage{enumerate}

\newcommand{\Amc}{\mathcal{A}}
\newcommand{\Bmc}{\mathcal{B}}
\newcommand{\Cmc}{\mathcal{C}}
\newcommand{\Gmc}{\mathcal{G}}
\newcommand{\Kmc}{\mathcal{K}}
\newcommand{\Lmc}{\mathcal{L}}

\newcommand{\ENmc}{\mathcal{N}}
\newcommand{\ENegmc}{\overline{\mathcal{N}}}

\newcommand{\Qmc}{\mathcal{Q}}
\newcommand{\Smc}{\mathcal{S}}
\newcommand{\Tmc}{\mathcal{T}}

\newcommand{\EZmc}{\mathcal{Z}}

\newcommand{\EXPTIME}{\mathsf{EXPTIME}}

\newcommand{\All}{\mathsf{A}}

\newcommand{\Bool}{\mathsf{Bool}}
\newcommand{\Bound}{\mathsf{B}}
\newcommand{\CTL}{\mathsf{CTL}}

\newcommand{\Ex}{\mathsf{E}}
\newcommand{\homom}{\preceq}
\newcommand{\neX}{\mathsf{X}}
\newcommand{\Prop}{\mathsf{P}}
\newcommand{\Rel}{\mathsf{R}}
\newcommand{\SAT}{\mathsf{SATCTL}^*}
\newcommand{\Until}{\mathsf{U}}

\newcommand{\Var}{\mathsf{V}}
\newcommand{\MSO}{\ensuremath{\mathsf{MSO}}}
\newcommand{\WMSO}{\ensuremath{\mathsf{WMSO}}}
\newcommand{\WMSOB}{\ensuremath{\mathsf{WMSO\!+\!B}}}
\newcommand{\defeq}{\mathrel{\mathop:}=}
\newcommand{\SNN}{S_{n-1}}

\newcommand{\reach}{\mathsf{reach}}

\newcommand{\arity}[1]{\mathsf{ar}({#1})}
\newcommand{\EHomDef}{{\sf EHomDef}}

\newcommand{\snnf}[1]{\hat{#1}}
\newcommand{\ExistsCycle}{\mathsf{ECycle}}
\newcommand{\Path}{\mathsf{Path}}
\newcommand{\BoundedPaths}{\mathsf{BPaths}}

\renewcommand{\phi}{\varphi}

\begin{document}

\title{Satisfiability of CTL${}^*$ with constraints\thanks{This work is supported by the DFG Research Training
    Group 1763 (QuantLA). The second author is supported by the DFG
    research project GELO.}}

\author{Claudia Carapelle \and  Alexander Kartzow \and Markus Lohrey}

\institute{Institut f\"ur Informatik, Universit\"at Leipzig, Germany}

\maketitle

\begin{abstract}
We show that satisfiability for $\CTL^*$ with equality-, order-, and mod\-ulo-constraints over $\mathbb{Z}$ is decidable.
Previously, decidability was only known for certain fragments of
$\CTL^*$, e.g., the existential and positive fragments and {\sf EF}.
\end{abstract}

\section{Introduction}

Temporal logics like $\mathsf{LTL}$, $\CTL$ or $\CTL^*$ are nowadays
standard languages for specifying system properties in
model-checking. They are interpreted over node labeled graphs (Kripke structures), where
the node labels (also called atomic propositions) represent abstract
properties of a system. Clearly, such an abstracted system state does
in  general not contain all the information of the original system state.
Consider for instance a program that manipulates two integer variables
$x$ and $y$.
A useful abstraction might be to introduce atomic
propositions $v_{-2^{32}}, \ldots, v_{2^{32}}$ for $v \in \{x,y\}$,
where the meaning of $v_k$ for $-2^{32} < k < 2^{32}$ is that the
variable $v\in \{x,y\}$ currently holds the value $k$,  and 
 $v_{-2^{32}}$  (resp., $v_{2^{32}}$) means that 
the current value of $v$ is at most $-2^{32}$ (resp., at least $2^{32}$). It is evident that such
an abstraction might lead to incorrect results in model-checking.

To overcome these problems, extensions of temporal logics with constraints
have been studied. Let us explain the idea in the context of
$\mathsf{LTL}$. For a fixed relational structure $\Amc$ (typical
examples for $\Amc$ are number domains like the integers or rationals
extended with certain relations) one adds atomic formulas
of the form $r(\neX^{i_1} x_1, \ldots, \neX^{i_k} x_k)$
(so called constraints) to standard  $\mathsf{LTL}$.
Here, $r$ is (a name of) one of the relations of the
structure $\Amc$, $i_1, \ldots, i_k \geq 0$,
and $x_1, \ldots, x_k$ are variables that range over
the universe of $\Amc$. An $\mathsf{LTL}$-formula containing such
constraints is interpreted over (infinite) paths of a standard Kripke structure, where in addition
every node (state) associates with each of the variables $x_1, \ldots,
x_k$ an element of $\Amc$ (one can think of $\Amc$-registers attached
to the system states).  A constraint $r(\neX^{i_1} x_1, \ldots, \neX^{i_k}
x_k)$ holds in a path $s_0 \to s_1 \to s_2 \to \cdots$ if the tuple
$(a_1, \ldots, a_k)$, where $a_j$ is the value of variable $x_j$ at
state $s_{i_j}$, belongs to the $\Amc$-relation $r$. In this way, the values of variables at different
system states can be compared. In our example from the first
paragraph, one might choose for $\Amc$ the structure $(\mathbb{Z}, <,
=, (=_a)_{a \in \mathbb{Z}})$, where $=_a$ is the unary predicate that
only holds for $a$. This structure has infinitely many predicates,
which is not a problem; our main result will actually talk about an
expansion of $(\mathbb{Z}, <, =, (=_a)_{a \in \mathbb{Z}})$.
Then, one might for instance write down a formula
$(<\!\!(x, \neX^1 y)) \Until (=_{100}\!\!(y))$ which holds on a path if and only
if there is a point of time where variable $y$ holds the value $100$
and for all previous points of time $t$, the value of $x$ at time $t$
is strictly smaller than the value of $y$ at time $t+1$.

In \cite{DemriG08}, Demri and Gascon studied $\mathsf{LTL}$ extended with constraints
from a language $\text{IPC}^*$. If we disregard succinctness aspects,
these constraints are equivalent to constraints over the structure
\begin{equation}\label{zmc^m}
\EZmc = (\mathbb{Z}, <, =, (=_a)_{a \in \mathbb{Z}},
(\equiv_{a,b})_{0 \leq a < b}),
\end{equation}
where $=_a$ denotes the
unary relation $\{a\}$ and $\equiv_{a,b}$ denotes the unary
relation $\{a+xb \mid x \in \mathbb{Z}\}$ (expressing that an integer
is congruent to $a$ modulo $b$).
The main result from \cite{DemriG08} states that satisfiability of $\mathsf{LTL}$
with constraints from $\EZmc$ is decidable and in fact
$\mathsf{PSPACE}$-complete, and hence has the same complexity as
satisfiability for $\mathsf{LTL}$ without constraints.
We should remark that the $\mathsf{PSPACE}$ upper bound from
\cite{DemriG08} even holds for the succinct $\text{IPC}^*$-representation of
constraints used in \cite{DemriG08}.

In the same way as outlined for $\mathsf{LTL}$ above, constraints
can be also added to $\CTL$ and $\CTL^*$ (then, constraints
$r(\neX^{i_1} x_1, \ldots, \neX^{i_k}x_k)$ are path formulas). 
A weak form of $\CTL^*$ with constraints from $\EZmc$ (where only 
integer variables and the same state can be compared) was
first introduced in \cite{Cerans00},
where it is used to describe properties of infinite transition systems, 
represented by relational automata. There it is shown that the model checking problem for
$\CTL^*$  over relational automata is undecidable. 

Demri and Gascon \cite{DemriG08} asked whether satisfiability of
$\CTL^*$ with constraints from $\EZmc$ over Kripke structures is decidable. 
This problem was investigated in \cite{BozzelliG06,Gascon09}, 
where several partial results where shown:
If we replace in $\EZmc$ the binary predicate $<$ by
unary predicates $<_c \; = \{ x \mid x < c \}$ for $c \in \mathbb{Z}$, then satisfiability
for $\CTL^*$ is decidable by \cite{Gascon09}. While, for the
full structure $\EZmc$ satisfiability is decidable for the
$\CTL^*$ fragment $\mathsf{CEF}^+$ (which contains the existential
and universal fragment of $\CTL^*$ as well as $\mathsf{EF}$) \cite{BozzelliG06}.

In this paper we prove that $\CTL^*$ with constraints over
$\EZmc$ is decidable. Our proof is divided into two steps.
The first step provides a tool to prove decidability of $\CTL^*$
with constraints over any structure $\Amc$ over a countable
(finite or infinite) signature $\Smc$ (the structure $\Amc$ has
to satisfy the additional property that the complement of any of its
relations has to be definable in positive existential first-order
logic over $\Amc$).
Let $\Lmc$ be a logic that satisfies the following two properties:
(i)  satisfiability of a given $\Lmc$-sentence
over the class of infinite node-labeled trees is decidable, and
(ii) $\Lmc$ is closed under boolean combinations with monadic
second-order formulas (MSO).
A typical such logic is MSO itself. By Rabin's seminal tree theorem
\cite{Rab69}, satisfiability of 
MSO-sentences over infinite node-labeled trees is decidable.
Assuming $\Lmc$ has these two properties,
we prove that satisfiability of $\CTL^*$ with constraints over $\Amc$
is decidable if one can compute from a given
finite subsignature $\sigma \subseteq \Smc$ an $\Lmc$-sentence $\psi_\sigma$ 
(over the signature $\sigma$) such that  for every countable $\sigma$-structure $\Bmc$:
$\Bmc \models \psi_\sigma$ if and only if
there exists a homomorphism from $\Bmc$ to $\Amc$ (i.e., a
  mapping from the domain of $\Bmc$ to the domain of $\Amc$ that
  preserves all relations from $\sigma$).
We say that the structure $\Amc$ has the property
\EHomDef($\Lmc$) if such a computable function $\sigma \mapsto
\psi_\sigma$ exists. 
\EHomDef  stands for ``existence of homomorphism is
definable''.  For instance, the structure $(\mathbb{Q}, <, =)$ has the property
\EHomDef($\MSO$), see Example~\ref{ex2}.

It is not clear whether  $\EZmc$ from \eqref{zmc^m} has
the property \EHomDef(\MSO) (we conjecture that it does not).
Hence, we need a different logic. It turns out that $\EZmc$
has the property \EHomDef(\WMSOB), where $\WMSOB$ is the
extension of weak monadic second-order logic (where only
quantification over finite subsets is allowed) with the bounding
quantifier $\mathsf{B}$.  A formula $\Bound X\colon\phi$ holds in
a structure $\Amc$  if and only if there exists a bound $b
\in \mathbb{N}$ such that for every finite subset $B$ of the domain of
$\Amc$ with $\Amc \models \phi(B)$ we have $|B| \leq b$.
Recently, Boja\'nczyk and Toru\'nczyk have shown that satisfiability
of $\WMSOB$ over infinite node-labeled trees is decidable \cite{BojanczykT12}.
The next problem is that $\WMSOB$ is not closed under boolean
combinations with $\MSO$-sentences. But fortunately, the decidability
proof for $\WMSOB$ can be extended to boolean combinations of
$\MSO$-sentences and ($\WMSOB$)-sentences, see Section~\ref{sec-mso} for details.
This finally shows that satisfiability of $\CTL^*$ with constraints
from $\EZmc$ is decidable.

While it would be extremely useful to add successor constraints ($y=x+1$) to $\EZmc$, 
this would lead to undecidability even for $\mathsf{LTL}$ \cite{DD07} and the very basic description logic  $\mathcal{ALC}$ \cite{Lutz04tocl},
which is basically multi-modal logic.
Nonetheless $\EZmc$ allows qualitative representation of increment, for example $x = y+1$ can be abstracted by 
$(y>x) \wedge ( \equiv_{1,2^k} \! (y))$ where $k$ is a large natural number. 
This is why temporal logics extended with constraints over $\EZmc$ seem to be a good compromise between (unexpressive) total abstraction and (undecidable) high concretion.

In the area of knowledge representation, extensions of description
logics with constraints from so called concrete domains have been
intensively studied, see \cite{Lutz02} for a survey. In \cite{Lutz2004235}, it was
shown that the extension of the description logic $\mathcal{ALC}$  with constraints from $(\mathbb{Q},<,=)$
has a decidable ($\EXPTIME$-complete) satisfiability problem with respect to general TBoxes
(also known as general concept inclusions). Such a TBox can be seen as
a second $\mathcal{ALC}$-formula that has to hold in all nodes of a model.
Our decidability proof is partly inspired by the construction from
\cite{Lutz2004235}, which in contrast to our proof is purely automata-theoretic.
Further results for description logics and concrete domains can be
found in \cite{Lutz04tocl,LutzM07}.

Unfortunately, our proof does not yield any complexity bound
for satisfiability of $\CTL^*$ with constraints from $\mathcal{Z}$.
The boolean combinations of ($\WMSOB$)-sentences and $\MSO$ sentences
that have to be checked for satisfiability (over infinite trees)
are of a simple structure, in particular their quantifier 
depth is not high. But no complexity statement for satisfiability of
$\WMSOB$ is made in \cite{BojanczykT12}, and it seems to be difficult
to analyze the algorithm from  \cite{BojanczykT12} (but it seems to be
elementary for a fixed quantifier depth). It is based on a
construction for cost functions over finite trees from \cite{ColcombetL10}, where the
authors only note that their construction seems to have very high complexity.

\section{Preliminaries}

Let $[1,d] = \{1,\ldots,d\}$.
For a word $w = a_1a_2 \cdots a_l \in [1,d]^*$ and $k \leq l$ we define
$w[:k] = a_1 a_2 \cdots a_k$;  it is the prefix of $w$ of length $k$.

Let $\Prop$ be a countable set of (atomic) propositions.
A Kripke structure over $\Prop$ is a triple $\Kmc = (D,\to,\rho)$, where
(i) $D$ is an arbitrary set of nodes (or states), 
(ii) $\to$ is a binary relation on $D$ such that
for every $u \in D$ there exists $v \in D$ with $u \to v$, and 
(iii)
$\rho : D \to 2^{\Prop}$ assigns to every node the set of propositions
that hold in the node. We require that $\bigcup_{v \in D} \rho(v)$ is
finite,
i.e., only finitely many propositions appear in $\Kmc$.
A {\em $\Kmc$-path} is an infinite sequence $\pi = (v_0, v_1, v_2, \ldots)$
such that $v_i \to v_{i+1}$ for all $i \geq 0$. For $i \geq 0$ 
we define the state $\pi(i) = v_i$ and the path $\pi^i =(v_i, v_{i+1},
v_{i+2},\ldots)$.
A {\em Kripke $d$-tree} is a Kripke structure of the form
$\Kmc = ([1,d]^*, \to, \rho)$, where $\to$ contains all pairs
$(u,ui)$ with $u \in [1,d]^*$ and $1 \leq i \leq d$, i.e.,
$([1,d]^*, \to)$ is a tree with root
$\varepsilon$ where every node has $d$ children.

A signature is a countable (finite or infinite) set 
$\Smc$ of relation symbols.
Every relation symbol $r \in \Smc$ has an associated
arity $\arity{r} \geq 1$. 
An $\Smc$-structure is a pair
$\Amc = (A,I)$, where $A$ is a non-empty set and $I$ maps every $r \in  \Smc$
to an $\arity{r}$-ary relation over $A$. Quite often, we will identify
the relation $I(r)$ with the relation symbol $r$, and we will 
specify an  $\Smc$-structure as $(A, r_1, r_2,\ldots)$ where
$\Smc = \{r_1,r_2,\ldots\}$. 
The $\Smc$-structure $\Amc = (A,I)$ is {\em negation-closed} if 
there exists a computable function that maps a relation symbol
$r \in \Smc$ to a positive existential first-order formula
$\varphi_r(x_1,\ldots,x_{\arity{r}})$ (i.e., a formula that is built up from atomic formulas
using $\wedge$, $\vee$, and $\exists$) 
such that $A^{\arity{r}}\setminus I(r) = \{ (a_1,\ldots,a_{\arity{r}}) \mid \Amc
\models \varphi_r(a_1,\ldots,a_{\arity{r}}) \}$. In other words, the 
complement of every relation $I(r)$ must be effectively definable 
by a positive existential first-order formula.

\begin{example}\label{ex1}
The structure $\EZmc$ from \eqref{zmc^m}
is negation-closed (we will write $x=a$ instead of $=_a\!\!(x)$ and similarly
for $\equiv_{a,b}$).
We have for instance:
\begin{itemize}
\item $x \neq y$ if and only if $x<y$ or $y<x$.
\item $x \neq a$ if and only if $\exists y \in \mathbb{Z} : y=a \wedge (x <
  y  \vee y < x)$.
\item $x \not\equiv a \text{ mod } b$ if and only if
$x \equiv c \text{ mod } b$ for some $0 \leq c < b$ with $a \neq c$.
\end{itemize}
\end{example} 
For a subsignature $\sigma \subseteq \Smc$, 
a $\sigma$-structure $\Bmc= (B,J)$ and an $\Smc$-structure $\Amc = (A,I)$,
a {\em homomorphism} $h :\Bmc \to \Amc$ is a mapping 
$h : B \to A$ such that for all $r \in \sigma$
and all tuples $(b_1,\ldots,b_{\arity{r}}) \in J(r)$ we have
$(h(b_1),\ldots,h(b_{\arity{r}})) \in I(r)$.
We write $\Bmc \homom \Amc$ if there is a homomorphism from $\Bmc$ to $\Amc$.

\section{$\MSO$ and  $\WMSOB$}\label{sec-mso}

Recall that {\em monadic second-order logic} ($\MSO$) is the extension
of first-order logic where also quantification over subsets of the underlying
structure is allowed.
We assume that the reader has some familiarity with $\MSO$.
{\em Weak monadic second-order logic} ($\WMSO$) has the same syntax
as $\MSO$ but second-order variables only range over finite subsets
of the underlying structure. Finally,  $\WMSOB$ is the extension of 
$\mathsf{WMSO}$ by the additional quantifier $\Bound X: \phi$
(the {\em bounding quantifier}). The semantics of $\Bound X: \phi$
in the structure $\Amc = (A,I)$ is defined as follows:
$\Amc \models \Bound X: \phi(X)$ if and only if there is a bound $b
\in \mathbb{N}$ such that $|B| \leq b$
for every finite subset $B \subseteq A$ with
$\Amc \models \phi(B)$.

\begin{example} \label{exa:WMSOBFormulas}
  For later use, we state some example formulas. 
  Let $\varphi(x,y)$ be a $\WMSO$-formula with two free first-order variables $x$ and $y$.
  Let $\Amc = (A,I)$ be a structure and let $E = \{ (a,b) \in A \times A \mid \Amc \models \varphi(a,b)\}$
  be the binary relation defined by $\varphi(x,y)$.
 We define the $\WMSO$-formula $\reach_{\varphi}(a,b)$ to be
  \begin{equation*}
    \exists X \;
      \forall Y \big(    a\in Y \land \forall x \forall y  ((x\in Y \land y\in
      X \land \varphi(x,y)) \rightarrow y\in Y) 
        \rightarrow b\in Y \big)      
  \end{equation*}
 It is straightforward to prove that
  $\Amc\models\reach_\varphi(a,b)$ if and only if $(a,b) \in E^*$.  
  Note that $\reach_{\varphi}$ is the standard $\MSO$-formula
  for reachability but restricted to some finite induced
  subgraph. Clearly, $b$ is reachable from $a$ in the graph
  $(A, E)$  if and  only if it is in some finite subgraph of $(A,E)$. 

  Let $\ExistsCycle_{\varphi}= \exists x\, \exists y  (\reach_\varphi(x,y) \land
  \varphi(y,x))$ be the \WMSO-formula expressing that there is a
  cycle in $(A,E)$. 
  
  Given a second-order variable $Z$, we define $\reach_\varphi^Z(a,b)$
  to be 
   \begin{equation*}
    a\in Z \land 
      \forall Y  \subseteq Z \big(    a\in Y \land \forall x \forall y
      ((x\in Y \land y\in  Z \land \varphi(x,y)) \rightarrow y\in Y) 
        \rightarrow b\in Y \big) .  
  \end{equation*}
  We have $\Amc \models \reach_\varphi^Z(a,b)$ iff $b$ is
  reachable from $a$  in the subgraph of $(A,E)$ induced by the
  (finite) set $Z$. Note that $\Amc\models \reach_\varphi^Z(a,b)$
  implies $\{a,b\}\subseteq Z$. 
  
  For the next examples we restrict our attention the case that the graph $(A,E)$ defined by $\varphi(x,y)$ is acyclic.
  Hence, the reflexive transitive closure $E^*$ is a partial order on $A$.
  Note that a finite set $F\subseteq A$ is
  an  $E$-path from $a\in F$ to $b\in F$ if and only if 
  $(F, (E \cap (F \times F))^*)$ is a finite linear order with all
  elements between $a$ and $b$.
  Define the \WMSO-formula $\Path_\varphi(a,b,Z)$ as
  \begin{equation*}
    \ \forall x \in  Z \; \forall y \in Z  \;
    (\reach_\varphi^Z(x,y) \vee \reach_\varphi^Z(y,x)) \ \wedge 
    \reach_\varphi^Z(a,x) \wedge \reach_\varphi^Z(x,b).
  \end{equation*}
  For every  acyclic $(A,E)$ we have
  $\Amc \models \Path_\varphi(a,b,P)$ if and only if $P$ contains
  exactly the nodes along  
  an $E$-path from $a$ to $b$. 

  We finally define the \WMSOB-formula
  $\BoundedPaths_\varphi(x,y) = \Bound Z: \Path_\varphi(x,y,Z)$.
  By definition of the quantifier $\Bound$, if $(A,E)$ is acyclic, then
  $\Amc \models \BoundedPaths_\varphi(a,b)$ if and only if there is a bound
  $k\in \mathbb{N}$ on the length of any  $E$-path from $a$ to $b$. 
\end{example}
Next, let $\Bool(\MSO,\WMSOB)$ be
the set of all Boolean combinations of $\MSO$-formulas and
($\WMSOB$)-formulas. We will use the following result.

\begin{theorem}[cf.~\cite{BojanczykT12}] \label{thm-bojan-to2}
One can decide whether  for a 
given  $d \in \mathbb{N}$  and a formula 
$\varphi \in \Bool(\MSO,\WMSOB)$ there exists a Kripke $d$-tree $\Kmc$
such that $\Kmc \models \varphi$.
\end{theorem} 

\begin{proof}
This theorem follows from results of Boja\'nczyk and
Toru\'nczyk \cite{BojanczykT12,BojanczykT12Long}.           
They introduced puzzles which can be seen as  pairs $P=(A, C)$, where
$A$ is a  parity tree automaton and $C$ is an unboundedness condition $C$ which
specifies a certain set of infinite paths labeled by states of $A$. 
A puzzle accepts a tree $\mathcal{T}$ if there is an accepting run $\rho$ of $A$
on $\mathcal{T}$ such that for each infinite path $\pi$ occurring in $\rho$,
$\pi\in C$ holds. In particular, ordinary parity tree automata can be seen
as puzzles with trivial unboundedness condition. The proof of our
theorem combines the following results.

\begin{lemma}[\cite{BojanczykT12}]\label{lem:WMSOBtoPuzzle}
  From a given ($\WMSOB$)-formula $\varphi$ and $d\in\mathbb{N}$ one can construct a puzzle
  $P_\varphi$ such that $\varphi$ is satisfied by some Kripke $d$-tree
  iff $P_\varphi$ is nonempty.
\end{lemma}

\begin{lemma}[\cite{BojanczykT12}] \label{lem:PuzzleEmptyness}
  Emptiness of puzzles is decidable.
\end{lemma}

\begin{lemma}[Lemma 17 of \cite{BojanczykT12Long}]\label{lem:PuzzlesIntersection}
  Puzzles are effectively closed under intersection.
\end{lemma}
Let $\phi\in \Bool(\MSO,\WMSOB)$. First,
$\phi$ can be effectively transformed into a disjunction $\bigvee_{i=1}^n
(\varphi_i \wedge \psi_i)$  where  $\varphi_i\in \MSO$
 and $\psi_i\in \WMSOB$ for all $i$.
By Lemma~\ref{lem:WMSOBtoPuzzle}, we can  construct a puzzle $P_i$ for $\psi_i$.
It is known that the $\MSO$-formula $\varphi_i$ can be translated into
a parity tree automaton $A_i$. Let $P'_i$ be a puzzle recognizing the
intersection of $P_i$ and $A_i$ (cf.~Lemma \ref{lem:PuzzlesIntersection}).
Now $\phi$ is satisfiable over Kripke $d$-trees if and only if
there is an $i$ such that $\varphi_i\land\psi_i$ is satisfiable over
Kripke $d$-trees if and only if there is an $i$ such that $P'_i$ is
nonempty.   By Lemma 
\ref{lem:PuzzleEmptyness}, the latter condition is decidable which
concludes the proof of the theorem. \qed
\end{proof}
Let $\Lmc$ be a logic (e.g. $\MSO$ or 
$\Bool(\MSO,\WMSOB)$). An $\Smc$-structure 
$\Amc$ has the property \EHomDef($\Lmc$)
(existence of homomorphisms to $\Amc$ is $\Lmc$-definable)
if there is a computable function that maps a finite subsignature
$\sigma \subseteq \Smc$ to an $\Lmc$-sentence $\varphi_\sigma$ such that for every
countable $\sigma$-structure $\Bmc$:
$\Bmc \homom \Amc$ if and only if $\Bmc \models \varphi_\sigma$.

\begin{example}\label{ex2}
The structure $\Qmc = (\mathbb{Q},<, =)$ has the property
\EHomDef($\WMSO$) (and \EHomDef(\MSO)).  In \cite{Lutz2004235} it is
implicitly shown that for a countable 
$\{<,=\}$-structure $\Bmc= (B,I)$, $\Bmc \homom \Qmc$ if and only if there does not
exist $(a,b) \in I(<)$ such that $(b,a) \in (I(<) \cup I(=) \cup I(=)^{-1})^*$.
This condition can be easily expressed in $\WMSO$ using the $\reach$-construction
from Example~\ref{exa:WMSOBFormulas}. Note that $I(=)$ is not required
to be the identity relation on $B$.
\end{example}

\section{$\CTL^*$ with constraints}

Let us fix a countably infinite set of atomic propositions $\Prop$
and a countably infinite set of variables $\Var$ for the rest of the paper.
Let $\Smc$ be a signature.
We define an extension of $\CTL^*$ with constraints over the
signature $\Smc$. We define $\CTL^*(\Smc)$-state formulas
$\phi$ and $\CTL^*(\Smc)$-path formulas
$\psi$ by the following grammar, where
$p\in\Prop$, $r \in \Smc$, 
$k = \arity{r}$, $i_1,\ldots,i_k \geq 0$, and $x_1, \ldots, x_k \in \Var$:
\begin{align*}
 \phi  ::= & p \mid \neg\phi \mid (\phi \wedge \phi) \mid \Ex\psi \\
 \psi  ::= & \phi \mid \neg\psi \mid (\psi \wedge \psi)  \mid \neX\psi
  \mid \psi \Until \psi \mid
  r(\neX^{i_1} x_1, \ldots, \neX^{i_k} x_k) 
\end{align*}
A formula of the form $R := r(\neX^{i_1} x_1, \ldots, \neX^{i_k} x_k)$
is also called an {\em atomic constraint} and we define 
$d(R) = \max\{i_1,\ldots,i_k\}$ (the depth of $R$). 
The syntactic difference between $\CTL^*(\Smc)$ and ordinary $\CTL^*$
lies in the presence of 
atomic constraints. 

Formulas of $\CTL^*(\Smc)$ are interpreted over 
triples $\Cmc = (\Amc, \Kmc, \gamma)$, where $\Amc = (A,I)$
is an $\Smc$-structure (also called the {\em concrete domain}), 
$\Kmc = (D,\to,\rho)$ is a Kripke structure over
$\Prop$, and $\gamma : D \times \Var \to A$ 
assigns to every $(v,x) \in D\times \Var$ a value $\gamma(v,x)$ (the value 
of variable $x$ at node $v$). We call such a triple $\Cmc = 
(\Amc, \Kmc, \gamma)$ an {\em $\Amc$-constraint graph}.
An $\Amc$-constraint graph $\Cmc = (\Amc, \Kmc, \gamma)$
is an {\em $\Amc$-constraint $d$-tree} if $\Kmc$ is a Kripke $d$-tree.

We now define the semantics of $\CTL^*(\Smc)$. For 
an $\Amc$-constraint graph $\Cmc = (\Amc,\Kmc,\gamma)$ with
$\Amc = (A, I)$ and $\Kmc = (D,\to,\rho)$,
a state $v \in D$, a $\Kmc$-path $\pi$,
a state formula $\phi$, and a path formula $\psi$ we write
$(\Cmc, v) \models \phi$ if $\phi$ holds in  $(\Cmc, v)$ and 
$(\Cmc, \pi) \models \psi$ if $\psi$ holds in  $(\Cmc, \pi)$. This is
inductively defined as follows (for the boolean connectives $\neg$
and $\wedge$ the definitions are as usual and we omit them):  
\begin{itemize}
\item $(\Cmc, v) \models p$ iff $p \in \rho(v)$.
\item $(\Cmc, v) \models \Ex\psi$ iff there is a
  $\Kmc$-path $\pi$ with $\pi(0)=v$ and $(\Cmc, \pi) \models \psi$.
\item $(\Cmc, \pi) \models \phi$ iff $(\Cmc, \pi(0))
  \models \phi$.
\item $(\Cmc, \pi) \models \neX\psi$ iff $(\Cmc, \pi^1)\models \psi$.
\item $(\Cmc, \pi) \models \psi_1 \Until \psi_2$ iff there
  exists $i \geq 0$ such that $(\Cmc, \pi^i) \models \psi_2$ and for
  all $0 \leq j < i$ we have $(\Cmc, \pi^j) \models \psi_1$.
\item  $(\Cmc, \pi) \models r(\neX^{i_1} x_1, \ldots, \neX^{i_n} x_n)$
 iff $(\gamma(\pi(i_1),x_1), \ldots, \gamma(\pi(i_n),x_n)) \in I(r)$.
\end{itemize}
Note that the role of the concrete domain $\Amc$ and of the valuation function $\gamma$ is 
restricted to the semantic of atomic constraints. $\CTL^*$-formulas are interpreted over Kripke structures, 
and to obtain their semantics it is sufficient to replace $\Cmc$ by 
 $\Kmc$ in the rules above and to remove the last line.

We use the usual abbreviations: $\theta_1 \vee \theta_2 := \neg
(\neg\theta_1 \wedge \neg\theta_2)$ (for both state and path formulas),
$\All\psi := \neg \Ex\neg\psi$ (universal path quantifier),
$\psi_1 \Rel \psi_2 := \neg (\neg\psi_1 \Until \neg\psi_2)$
(the release operator). Note that $(\Cmc, \pi) \models \psi_1 \Rel
\psi_2$ iff ($(\Cmc, \pi^i) \models \psi_2$ for all 
$i \geq 0$  or there exists $i \geq 0$ such that $(\Cmc, \pi^i) \models \psi_1$ and $(\Cmc, \pi^j) \models \psi_2$ for all $0 \leq j \leq i$).

Using this extended set of operators we can 
put every formula into a semantically equivalent 
{\em negation normal form}, where $\neg$ only occurs in front of 
atomic propositions or  atomic constraints.
Let $\#_{\Ex}(\theta)$ be the the number of different subformulas of the form 
$\Ex \psi$ in the $\CTL^*(\Smc)$-formula
$\theta$. Then $\CTL^*(\Smc)$ has the following tree model property:

\begin{theorem}[cf.~\cite{Gascon09}] \label{thm-tree-model}
Let $\phi$ be a $\CTL^*(\Smc)$-state formula in negation normal form and let $\Amc=(A,I)$ be an $\Smc$-structure. 
Then $\phi$ is $\Amc$-satisfiable if and only if there exists
an {\em $\Amc$-constraint $(\#_{\Ex}(\phi)+1)$-tree} $\Cmc$ 
with $(\Cmc,\varepsilon) \models \phi$.
\end{theorem}
Note that for checking $(\Amc, \Kmc, \gamma) \models \phi$  we may ignore all
propositions $p \in \Prop$ that do not occur in $\phi$. Similarly,
only those values $\gamma(u,x)$, where $x$ is a variable that appears in $\phi$,
are relevant. Hence, if $\Var_\phi$ is the finite set of variables that occur in
$\phi$, then we can consider $\gamma$ as a mapping from $D \times \Var_\phi$ to
the domain of $\Amc$. Intuitively, we assign to each node $u \in D$ registers
that store the values $\gamma(u,x)$ for $x \in \Var_\phi$.

\section{Satisfiability of constraint $\CTL^*$ over a concrete domain}

When we talk about satisfiability for $\CTL^*(\Smc)$
our setting is as follows: We fix a concrete domain $\Amc=(A,I)$. Given
a $\CTL^*(\Smc)$-state formula $\phi$, we say that 
$\phi$ is $\Amc$-satisfiable if there
is an $\Amc$-constraint graph $\Cmc = (\Amc, \Kmc, \gamma)$ and a
node $v$ of $\Kmc$ such that $(\Cmc, v) \models \phi$.
With $\SAT(\Amc)$ we denote the following computational problem:
{\em Is a given state formula $\phi \in \CTL^*(\Smc)$ $\Amc$-satisfiable?}
The main result of this section is:

\begin{theorem} \label{thm-EHOMDef-SAT} 
Let $\Amc$ be a negation-closed $\Smc$-structure, which moreover has the 
property \EHomDef$(\Bool(\MSO,\WMSOB))$.
Then the problem $\SAT(\Amc)$ is decidable.
\end{theorem}
We say that a $\CTL^*(\Smc)$-formula $\varphi$
is in \emph{strong negation normal form} if negations
only occur  in front of atomic propositions (i.e., $\varphi$ is in
negation normal form and there is no subformula $\neg R$ where
$R$ is an atomic constraint).

Let us fix a  $\CTL^*(\Smc)$-state formula $\phi$ in negation normal
form and a negation-closed $\Smc$-structure $\Amc$ for the rest of
this section. We want to check whether $\phi$ is $\Amc$-satisfiable.
First, we reduce to formulas in strong negation normal form:

\begin{lemma} \label{lemma-neg-closed}
  Let $\Amc=(A,I)$ be a negation-closed $\Smc$-structure. 
  From a given  $\CTL^*(\Smc)$-state formula $\varphi$ one can compute a $\CTL^*(\Smc)$-state formula
  $\snnf{\varphi}$ in strong negation normal form such that 
  $\varphi$ is $\Amc$-satisfiable iff $\snnf{\varphi}$ is
  $\Amc$-satisfiable. 
\end{lemma}
\begin{proof}
    We can assume that $\varphi$ is
  in negation normal form.  Using induction, it suffices to eliminate a single
  negated atomic constraint $\theta = \neg r(\neX^{i_1} x_1, \dots, \neX^{i_k} x_k)$ in $\varphi$,
  where $k = \arity{r}$. Let  $d=\max\{i_1, \ldots, i_k\}$, which is the depth of the constraint
  $r(\neX^{i_1} x_1, \dots, \neX^{i_k} x_k)$. Since $\Amc$ is negation-closed, we can 
  compute a positive quantifier-free first-order formula 
  $\psi(y_1, y_2, \ldots, y_k, z_1, z_2, \dots, z_m)$ over the
  signature $\Smc$  such that
  \mbox{$\Amc\models  \neg r(a_1, \dots, a_k)$} if and only if
  $\Amc\models \exists z_1 \cdots \exists z_m  \, \psi(a_1, \dots, a_k,
  z_1, \ldots, z_m)$. 
  Let $y'_1, \dots, y'_m$ be fresh variables not occurring in
  $\varphi$.  We define the $\CTL^*(\Smc)$-state formula $\varphi'$
  by  replacing in $\varphi$
  every occurrence of the negated constraint $\theta$ by
  the path formula
  $$
  \psi(\neX^{i_1} x_1, \ldots, \neX^{i_k} x_k, \neX^d y'_1, \ldots, \neX^d y'_m) .
  $$
  So, we replace in $\psi(y_1, \ldots, y_k, z_1, \dots, z_m)$ every occurrence of a variable $y_p$ (resp., $z_q$) by
  $\neX^{i_p} x_p$ (resp., $\neX^d y'_q$).
  
  We first prove that
  $\varphi'$ is $\Amc$-satisfiable  if 
  $\varphi$ is $\Amc$-satisfiable.
  If $\varphi$ is $\Amc$-satisfiable, then by Thm.~\ref{thm-tree-model}
  there is an
  $\Amc$-constraint $t$-tree $\Cmc = (\Amc, \Kmc, \gamma)$ with $(\Cmc,\varepsilon) \models \varphi$,
  where $\Kmc = ([1,t]^*,\to,\rho)$ and $\gamma$ has domain
  $[1,t]^*\times \Var_\varphi$ for $\Var_\varphi$ the set of variables of $\varphi$. 
  By choice of the fresh variables, we have $\Var_\varphi \cap \{y'_1, \ldots, y'_m\} = \emptyset$.
  Now we extend $\gamma$ to 
  $\gamma': [1,t]^*\times(\Var_\varphi \cup \{y'_1, \ldots, y'_m\}) \to A$ as follows:
  Consider $w,v\in [1,t]^*$ such that
  $\lvert v \rvert=d$ and  let $\pi$ be a path in the tree $([1,t]^*,\to)$ starting
  at $w$ and passing $wv$, i.e., $\pi(0)=w$ and $\pi(d)=wv$.  Let $v_p = v[:i_p]$ for $1 \leq p \leq k$.
  \begin{itemize}
  \item   If $(\Kmc, \pi) \models 
    \theta = \neg r(\neX^{i_1} x_1, \dots, \neX^{i_k} x_k)$ then there
    are values $a_1, \dots, a_m\in A$ such that
    $\Amc\models
    \psi(\gamma(wv_1, x_1),\ldots,\gamma(wv_k, x_k),  a_1, \dots, a_m)$. 
    Note that the choice of $a_1, \dots, a_m$ can be made independent of the
    concrete choice of $\pi$ but only depending on
    $\gamma$ and $wv$. Thus, it is well-defined to set 
    $\gamma'(wv,y'_q) = a_q$ for all $1\leq q \leq m$. 
  \item If $(\Kmc, \pi) \models 
     r(\neX^{i_1} x_{1}, \ldots, \neX^{i_k} x_{k})$, then we choose
    $\gamma'(wv,y'_q)\in A$ arbitrarily.
  \end{itemize}
  Finally, for all $w$ such that $\lvert w \rvert < d$ we choose
  $\gamma'(w,y'_q)\in A$ arbitrarily.

  By induction on the structure of $\varphi$ we prove that for 
  $\Cmc'=(\Amc,\Kmc, \gamma')$ we have $(\Cmc',\varepsilon) \models \varphi'$. All steps
  are trivial except for the case that the subformula is
  $\theta = \neg r(\neX^{i_1}x_1, \ldots \neX^{i_k}x_k)$. In this
  case we assume that  $(\Cmc, \pi) \models \theta$ for a path $\pi$, and we have to
  show that 
  $$
  (\Cmc',\pi) \models  \theta' = \psi(\neX^{i_1} x_1, \ldots, \neX^{i_k} x_k, \neX^d y'_1, \ldots, \neX^d y'_m).
  $$
  By definition $\pi(d)=wv$ for some word $w=\pi(0)$ and some
  word $v$ such that $\lvert v \rvert =d$. Let $v_p = v[:i_p]$ for $1 \leq p \leq k$.
  Since $(\Cmc,\pi) \models \neg r(\neX^{i_1} x_1, \dots, \neX^{i_k} x_k)$, we conclude
  immediately that 
  \begin{align*}
    \Amc \models
    \psi(\gamma(wv_1,x_1),\ldots,\gamma(wv_k),
    x_k),  \gamma'(wv,y'_1), \ldots, \gamma'(wv,y'_m)).    
  \end{align*}
  Noting that $w (v[:d]) = wv$ we immediately conclude that
  $(\Cmc',\pi) \models \theta'$ which concludes the first
  direction. 

  In order to prove that $\varphi$ is $\Amc$-satisfiable if
  $\varphi'$ is $\Amc$-satisfiable, let us assume (using again Thm.~\ref{thm-tree-model})
  that $\Cmc'=(\Amc,\Kmc,\gamma')$ is an $\Amc$-constraint 
  $t$-tree such that $(\Cmc',\varepsilon) \models \varphi'$. 
  Let $\Cmc$ be the $\Amc$-constraint $t$-tree obtained from $\Cmc'$
  by restricting $\gamma'$ to the variables from $\Var_\varphi$.
  Again by induction on the structure of $\varphi$, we end
  up with the task to show that if $(\Cmc',\pi)\models \theta'$ for a path $\pi$,
  then $(\Cmc,\pi) \models \theta$.
  If 
  $$
  (\Cmc',\pi)\models \theta' = \psi(\neX^{i_1} x_1, \ldots, \neX^{i_k} x_k, \neX^d y'_1, \ldots, \neX^d y'_m),
  $$
  then there are values (namely, 
  $\gamma'(\pi(d),y'_1), \ldots, \gamma'(\pi(d),y'_m)$)
  witnessing
  \begin{align*}
    \Amc\models \exists z_1\cdots \exists z_m 
    \psi( \gamma( \pi(i_1),x_1), \ldots, \gamma(\pi(i_k),x_k), z_1, \ldots, z_m).    
  \end{align*}
  By choice of $\psi$ this implies that
  $\Amc\models \neg r(\gamma( \pi(i_1),x_1), \ldots, \gamma(\pi(i_k),x_k))$. Hence, we have
  $(\Cmc, \pi) \models 
  \neg r(\neX^{i_1} x_1, \dots, \neX^{i_k} x_k) = \theta$. 
\qed
\end{proof}
From now on let us assume that $\phi$ is in strong negation normal
form. 
Let $d = \#_{\Ex}(\phi) +1$.
Let $R_1,\ldots,R_n$ be a list of all atomic constraints that are
subformulas of $\phi$, and let $\Var_\varphi$ be the finite set of variables
that occur in $\varphi$.
Let us fix new propositions $p_1,\ldots,p_n$ (one for each $R_i$) that do not
occur in $\phi$. Let $d_i = d(R_i)$ be the depth of the constraint $R_i$.
We denote with $\phi^a$ the (ordinary)
$\CTL^*$-formula obtained from $\phi$ by replacing every occurrence
of a constraint $R_i$ by $\neX^{d_i} p_i$.
Given an $\Amc$-constraint $d$-tree $\Cmc = (\Amc, \Kmc, \gamma)$,
where $\Kmc = ([1,d]^*,\to,\rho)$ and $\rho(v) \cap
\{p_1,\ldots,p_n\}=\emptyset$ for all $v \in [1,d]^*$,
 we define a Kripke $d$-tree $\Cmc^a = ([1,d]^*,\to,\rho^a)$, where
$\rho^a(v)$ contains
\begin{itemize}
\item all propositions from $\rho(v)$ and
\item all propositions $p_i$ ($1 \leq i \leq n$) such that 
the following holds, where we assume that $R_i$ has the form 
$r(\neX^{j_1}x_1,\ldots,\neX^{j_k}x_k)$ with $k = \arity{r}$ 
(hence, $d_i = \max\{j_1,\ldots,j_k\}$):
\begin{itemize}
\item$v = su$ with $|u| = d_i$
\item$( \gamma(su_1,x_1), \ldots, \gamma(su_k,x_k)) \in I(r)$, where $u_l = u[:j_l]$ for $1 \leq l \leq k$.
\end{itemize}
\end{itemize}
Hence, the fact that proposition $p_i$ labels node $su$ with $|u|=d_i$
means that the constraint $R_i$
holds along every path that starts in node $s$ and descends in the tree
down via node $su$. The superscript ``$a$'' in $\Cmc^a$ stands for
``abstracted'' since we abstract from the concrete constraints and
replace them by new propositions.

Moreover, given a Kripke $d$-tree $\Tmc = ([1,d]^*,\to,\rho)$ 
(where the new propositions $p_1,\ldots, p_n$ are allowed to occur in $\Tmc$)
we define a countable $\Smc$-structure $\Gmc_{\Tmc} = ([1,d]^*\times\Var_\varphi, J)$ 
as follows:
The interpretation $J(r)$ of the 
relation symbol $r \in \Smc$ contains all $k$-tuples (where $k = \arity{r}$)
$((su_1,x_1), \ldots, (su_k, x_k))$ for which
there exist $1 \leq i \leq n$ and $u \in [1,d]^*$ with $|u|=d_i$
such that $p_i \in \rho(su)$,
$R_i = r(\neX^{j_1}x_1,\ldots,\neX^{j_k}x_k)$,
and $u_t = u[:j_t]$ for $1 \leq t \leq k$.

\begin{figure}[t]
\begin{center}
  \setlength{\unitlength}{1mm}
  \begin{picture}(100,35)(0,40)
    \gasset{Nframe=n,Nfill=y,AHnb=0,ELdist=0.4}
    \gasset{Nw=1.2,Nh=1.2,ExtNL=y,NLangle=90,NLdist=.7,dash={0.2 0.5}0}
    \node(eps)(50,70){}
   
    \node[NLangle=100](a)(20,60){$p_1 p_2$}
    \node[NLangle=-90](b)(80,60){$p_1$}
    \drawedge[ELside=r](eps,a){}
    \drawedge[ELside=r](eps,b){}

    \node(aa)(5,50){$p_1 p_2$}
    \node(ab)(35,50){}
    \drawedge[ELside=r](a,aa){}
    \drawedge[ELside=r](a,ab){}
    
    \node(ba)(65,50){}
    \node(bb)(95,50){}
    \drawedge[ELside=r](b,ba){}
    \drawedge[ELside=r](b,bb){}
    
    \node(aaa)(-3,42){$p_1$}
    \node(aab)(13,42){$p_2$}
    \drawedge[ELside=r](aa,aaa){}
    \drawedge[ELside=r](aa,aab){}
    
    \node(aba)(27,42){}
    \node[NLangle=-90](abb)(43,42){$p_1$}
    \drawedge[ELside=r](ab,aba){}
    \drawedge[ELside=r](ab,abb){}
    
    \node(baa)(57,42){}
    \node(bab)(73,42){$p_2$}
    \drawedge[ELside=r](ba,baa){}
    \drawedge[ELside=r](ba,bab){}

    \node[NLangle=120](bba)(87,42){$p_1 p_2$}
    \node(bbb)(103,42){$p_2$}
    \drawedge[ELside=r](bb,bba){}
    \drawedge[ELside=r](bb,bbb){}
    
    \gasset{Nw=.7,Nh=.7,NLangle=0,NLdist=.7}
     \node[NLangle=90](eps-x)(46,70){$1$}
    \node[NLangle=90](eps-y)(54,70){$2$}
    
    \node[NLangle=180](ax)(16,60){$2$}
    \node(ay)(24,60){$2$}
    \node[NLangle=180](bx)(76,60){$1$}
    \node(by)(84,60){$3$}

    \node[NLangle=180](aax)(1,50){$3$}
    \node(aay)(9,50){$3$}
    \node[NLangle=180](abx)(31,50){$2$}
    \node(aby)(39,50){$0$}
    
    \node[NLangle=180](bax)(61,50){$2$}
    \node(bay)(69,50){$0$}
    \node[NLangle=180](bbx)(91,50){$3$}
    \node(bby)(99,50){$0$}
    
    \node[NLangle=270](aaax)(-7,42){$0$}
    \node[NLangle=270](aaay)(1,42){$4$}
    \node[NLangle=270](aabx)(9,42){$2$}
    \node[NLangle=270](aaby)(17,42){$2$}
    
    \node[NLangle=270](abax)(23,42){$0$}
    \node[NLangle=270](abay)(31,42){$2$}
    \node[NLangle=270](abbx)(39,42){$0$}
     \node[NLangle=270](abby)(47,42){$3$}
     
    \node[NLangle=270](baax)(53,42){$0$}
    \node[NLangle=270](baay)(61,42){$2$}
    \node[NLangle=270](babx)(69,42){$0$}
    \node[NLangle=270](baby)(77,42){$0$}
    
    \node[NLangle=270](bbax)(83,42){$4$}
    \node[NLangle=270](bbay)(91,42){$4$}
    \node[NLangle=270](bbbx)(99,42){$3$}
    \node[NLangle=270](bbby)(107,42){$3$}  
     
   \gasset{curvedepth=-2,ELside=r,AHnb=1,dash={}{0}}
    \drawedge(ax,ay){$=$}
    \drawedge(aax,aay){$=$}
    \drawedge(aabx,aaby){$=$}
    \drawedge(babx,baby){$=$}
    \drawedge(bbax,bbay){$=$}
    \drawedge(bbbx,bbby){$=$}
    
    \drawedge[curvedepth=-3](eps-x,ay){$<$} 
    \drawbpedge(eps-x,-30,20,by,150,20){$<$} 
    \drawedge[curvedepth=0](ax,aay){$<$} 
    \drawedge[curvedepth=0](aax,aaay){$<$} 
    \drawbpedge(abx,-45,5,abby,135,5){$<$} 
    \drawedge[curvedepth=0,ELside=l](bbx,bbay){$<$} 
    \end{picture}
\end{center}
\caption{\label{fig-constraint-graph}The $(\mathbb{N},<,=)$-constraint 2-tree
$\Cmc$ from Ex.~\ref{ex-constraint-graph}, the Kripke 2-tree $\Tmc = \Cmc^a$, and 
the structure $\Gmc_{\Tmc}$.}
\end{figure}

\begin{example} \label{ex-constraint-graph}
Figure~\ref{fig-constraint-graph} shows an example, where we 
assume that $d=2$ and $n=2$, $R_1= [ <\!\!(x_1, \neX x_2) ]$, and
$R_2 = [ =\!\!(\neX x_1, \neX x_2) ]$. The figure shows an initial
part of an $(\mathbb{N},<,=)$-constraint 2-tree
$\Cmc = ((\mathbb{N},<,=), \Kmc, \gamma)$.
The edges of the Kripke $2$-tree $\Kmc$ are dotted.  We assume that $\Kmc$ is defined over the empty set of 
propositions. The node to the left (resp., right) of a tree node $u$ is labeled
by the value $\gamma(u,x_1)$ (resp.~$\gamma(u,x_2)$).
The figure shows the labeling of tree nodes with the two new propositions $p_1$ and $p_2$ (corresponding to 
$R_1$ and $R_2$) as well as the $\{<,=\}$-structure $\Gmc_{\Tmc}$ for $\Tmc = \Cmc^a$.
\end{example}

\begin{lemma} \label{lemma-first-reduction}
Let $\phi$ be a $\CTL^*(\Smc)$-state formula in strong negation normal form.
The formula $\phi$ is $\Amc$-satisfiable if and only if there exists
a Kripke $(\#_{\Ex}(\phi) +1)$-tree $\Tmc$ such that $(\Tmc,\varepsilon)
\models \phi^a$ and $\Gmc_{\Tmc} \homom \Amc$.
\end{lemma}
\begin{proof}
Let us first assume that $\phi$ is $\Amc$-satisfiable and let
$\Cmc = (\Amc, \Kmc, \gamma)$ be an $\Amc$-constraint graph with
$\Amc=(A,I)$
and $v$ a node of $\Kmc$ such that $(\Cmc, v) \models \phi$.
By Thm.~\ref{thm-tree-model} we can assume that $\Kmc = ([1,d]^*,\to,\rho)$ is
a Kripke $d$-tree with $d=e+1$ and $v = \varepsilon$.
Let $m$, $n$, $R_i$ and $d_i$ ($1 \leq i \leq n$) have the same
meaning as above.
Take the Kripke $d$-tree $\Tmc = \Cmc^a = ([1,d]^*,\to,\rho^a)$.
We claim that $\gamma : [1,d]^* \times \Var_\varphi \to A$
is a homomorphism from $\Gmc_{\Tmc}$ to $\Amc$. For this,
assume that $((su_1,x_1), \ldots, (su_k, x_k))$ belongs to the
interpretation of $r$ in $\Gmc_{\Tmc}$.
Hence, there exist $1 \leq i \leq n$ and $u \in [1,d]^*$ with $|u|=d_i$
such that $p_i \in \rho^a(su)$,
$R_i = r(\neX^{j_1}x_1,\ldots,\neX^{j_k}x_k)$,
and $u_q = u[:j_q]$ for $1 \leq q \leq k$.
Since $\Tmc = \Cmc^a$ and $p_i \in \rho^a(su)$, it follows that
the tuple
$(\gamma(su_1,x_1), \ldots, \gamma(su_k,x_k))$
belongs to the interpretation of $r$ in $\Amc$. Hence,
$\gamma$ is indeed a homomorphism.

In order to show $(\Tmc,\varepsilon)\models \phi^a$ we prove by
induction on the structure of formulas the following implication, where
$\psi$ is a state or path subformula of $\phi$, $v \in [1,d]^*$ is a
node and $\pi$ is a $\Kmc$-path (and hence also a $\Tmc$-path):
If $(\Cmc, v) \models \psi$, then $(\Tmc,v)\models \psi^a$, and if
$(\Cmc, \pi) \models \psi$ then $(\Tmc,\pi)\models \psi^a$.
\begin{itemize}
\item $\psi = p \in P$: We have $\psi^a = p$.
  If $v$ is such that $(\Cmc, v) \models p$, we have $p \in \rho (v)$
  and, since $\rho(v) \subseteq \rho^a(v)$,
  $(\Tmc,v)\models p$. If $\pi$ is a path such that $(\Cmc, \pi)
  \models p$, then 
  $(\Cmc, \pi(0)) \models p$. Using what we have just proven,
  $(\Tmc,\pi(0))\models p$ and thus $(\Tmc,\pi)\models p$.
\item $\psi = \neg p$ with $p \in P$ (recall that negations
  only occurs in front of atomic propositions): We have $\psi^a = \neg p$:
  If $v$ is such that $(\Cmc, v) \models \neg p$, we have $p \not\in
  \rho(v)$. Note that $p \not\in \{p_1,\ldots,p_n\}$. Since
  $\rho(v) = \rho^a \setminus \{p_1,\ldots,p_n\}$ we have
  $p \not\in \rho^a(v)$. Hence,  $(\Tmc, v) \models \neg p$.
  For a path $\pi$ with $(\Cmc, \pi) \models \neg p$ we can argue
  in the same way.
\item $\psi = R_i$ for some $1 \leq i \leq n$:
  Suppose that $R_i = r(\neX^{j_1}y_1,\ldots,\neX^{j_k}y_k)$
  where $d_i = \max\{j_1,\ldots,j_k\}$ is the depth of $R_i$.
  We have $\psi^a = \neX^{d_i} p_i$. Let $\pi$ be a path such that
  $(\Cmc, \pi) \models R_i$.
  By definition  $(\gamma(\pi(j_1),y_1),\ldots, \gamma(\pi(j_k),y_k))
  \in I(r)$ and 
  therefore $p_i \in \rho^a(\pi(d_i))$. This means that
  $(\Tmc,\pi^{d_i})\models p_i$ and
  consequently that $(\Tmc,\pi)\models \neX^{d_i}p_i$.
\item $\psi = \psi_1 \circ \psi_2$ for $\circ \in \{
  \wedge,\vee\}$ and state or path formulas $\psi_1$ and $\psi$:
  Then we have $\psi^a = \varphi_1^a \circ \varphi_2^a$, and we can
  directly argue by induction.
\item $\psi = \Ex \varphi$: We have $\psi^a = \Ex \varphi^a$.
  If $(\Cmc, v) \models \Ex\varphi$ then there must be a
  path $\pi$ with $\pi(0)=v$ and $(\Cmc, \pi) \models \varphi$. By
  induction, we have $(\Tmc, \pi) \models \varphi ^a$ and therefore
  $(\Tmc, v) \models \Ex\varphi^a $. The case $\psi = \All \varphi$ is
  treated similarly.  Moreover, the case that $\Ex \varphi$ or $\All
  \varphi$ is interpreted as a 
  path formula directly reduces to the case of a state formula.
\item  $\psi=\neX \varphi$: We have $\psi^a =  \neX
  \varphi^a$. Let $\pi$ be a path such that $(\Cmc, \pi) \models \neX
  \varphi$. Then $(\Cmc, \pi^1) \models \varphi$. By induction,
  $(\Tmc, \pi^1) \models \varphi^a$ and hence $(\Tmc, \pi)
  \models \neX \varphi^a $.
\item  $\psi=\varphi_1 \Until \varphi_2$: We have $\psi^a =
  \varphi_1^a \Until \varphi_2^a$. Let $\pi$ be a path such that
  $(\Cmc, \pi) \models \varphi_1 \Until \varphi_2$. Then there exists $i
  \geq 0$ such that $(\Cmc, \pi^i) \models \varphi_2$
  and $(\Cmc, \pi^j) \models \varphi_1$ for all $0\leq j<i$. By
  induction we obtain  $(\Tmc, \pi^i) \models \varphi_2^a$ and  $(\Tmc,
  \pi^j) \models \varphi_1^a$ for all $0\leq j<i$. From this we get
  $(\Tmc, \pi) \models \varphi_1^a \Until \varphi_2^a $.
\item  $\psi = \varphi_1\Rel \varphi_2$: We have $\psi^a =
  \varphi_1^a \Rel \varphi_2^a$. Let $\pi$ be a path such that $(\Cmc,
  \pi) \models \varphi_1 \Rel \varphi_2$. This means that
  $(\Cmc, \pi^i) \models\varphi_2$ for all $i \geq 0$, or there exists
  $i \geq 0$ such that  $(\Cmc, \pi^i) \models\varphi_1$ and
  $(\Cmc, \pi^j) \models\varphi_2$ for all $0 \leq j \leq i$.
  Again, using induction, we get: $(\Tmc, \pi^i) \models\varphi^a_2$
  for all $i \geq 0$, or there exists 
  $i \geq 0$ such that  $(\Tmc, \pi^i) \models\varphi^a_1$ and
  $(\Tmc, \pi^j) \models\varphi^a_2$ for all $0 \leq j \leq i$.
  But this means that $(\Tmc, \pi) \models \varphi_1^a \Rel \varphi_2^a $.
\end{itemize}
This concludes the proof of the ``only if'' direction from the lemma.
For the other direction, assume that there exists
a Kripke $d$-tree $\Tmc = ([1,d]^*, \rightarrow, \rho_{\Tmc})$ such that $(\Tmc,\varepsilon)
\models \phi^a$ and there exists a homomorphism $h$ from
$\Gmc_{\Tmc}$ to $\Amc$. Define the $\Amc$-constraint
graph $\Cmc = (\Amc, \Kmc, h)$, where  $\Kmc = ([1,d]^*, \rightarrow, \rho)$ with $\rho(v)
= \rho_{\Tmc} (v) \backslash \{p_1, \dots, p_n \}$ for all $v \in
[1,d]^*$. 
We claim that $(\Cmc,\varepsilon) \models \phi$.

Again, we can prove by induction that for all (state or path)
subformulas $\psi$ of $\phi$, for all $v \in [1,d]^*$, and for all
$\Tmc$-paths $\pi$, if $(\Tmc, v) \models \psi^a$ then $(\Cmc, v)
\models \psi$, and if $(\Tmc, \pi) \models \psi^a$ then $(\Cmc, \pi) \models \psi$.
The only nontrivial part is the case that $\psi$ is one of the atomic
constraints $R_i = r (X^{j_1}x_1, \dots, X^{j_k}x_k)$, where
$k=\arity{r}$. This means that $\psi ^a = X^{d_i} p_i$, where $d_i =
\max\{j_1,\ldots,j_k\}$ is the depth of $R_i$. If $\pi$ is such that
$(\Tmc, \pi) \models \psi^a$, this means that $p_i \in
\rho_{\Tmc}(\pi(d_i))$. Then, according to the definition of
$\Gmc_{\Tmc}$, the interpretation of $r$ in $\Gmc_{\Tmc}$ contains the
$k$-tuple $((\pi(j_1),x_1),\ldots,(\pi(j_k),x_k))$. Since $h$ is a
homomorphism from 
$\Gmc_{\Tmc}$ to $\Amc$, we have
$(h(\pi(j_1),x_1),\dots,h(\pi(j_k),x_k)) \in I(r)$. By definition of $\Cmc$ this means that
$(\Cmc,\pi) \models  r(X^{j_1}x_1, \ldots, X^{j_k}x_k)$.
\qed
\end{proof}
Let $\theta = \phi^a$ for the further discussion. Hence, $\theta$ 
is an ordinary $\CTL^*$-state formula, where negations only occur in front of 
propositions from $\Prop \setminus \{p_1,\ldots,p_m\}$,
and $d=\#_{\Ex}(\theta)+1$. By Lemma~\ref{lemma-first-reduction}, we have to check,
whether there exists a Kripke $d$-tree $\Tmc$ such
that $(\Tmc,\varepsilon)\models \theta$ and 
$\Gmc_{\Tmc} \homom \Amc$. 

Let $\sigma \subseteq \Smc$ be the finite subsignature consisting of
all predicate symbols that occur in our initial $\CTL^*(\Smc)$-formula
$\phi$. Note that $\Gmc_{\Tmc}$ is actually a $\sigma$-structure. 
Since the concrete domain $\Amc$ has the 
property \EHomDef$(\Bool(\MSO,\WMSOB))$, one can compute from $\sigma$
a $\Bool(\MSO,\WMSOB)$-formula $\alpha$ such that
for every countable $\sigma$-structure $\Bmc$ we have 
$\Bmc \models \alpha$ if and only if $\Bmc \homom\Amc$. Hence, our new goal is to decide,
whether there exists a Kripke $d$-tree $\Tmc$ such
that 
$(\Tmc,\varepsilon)\models \theta$ and $\Gmc_{\Tmc} \models \alpha$ (note that
$\Gmc_{\Tmc}$ is countable).
It is well known that every $\CTL^*$-state formula can be effectively transformed
into an equivalent $\MSO$-formula with a single 
free first-order variable. Since the root $\varepsilon$ of a tree is first-order definable,
we get an $\MSO$-sentence $\psi$ such that 
$(\Tmc,\varepsilon)\models \theta$ if and only if $\Tmc \models \psi$.
Hence, we have to check 
whether there exists a Kripke $d$-tree $\Tmc$ such
that 
$\Tmc\models \psi$ and $\Gmc_{\Tmc} \models \alpha$. 
If we can translate the $\Bool(\MSO,\WMSOB)$-formula $\alpha$
back into a $\Bool(\MSO,\WMSOB)$-formula $\alpha'$ such that
($\Gmc_{\Tmc} \models \alpha \Leftrightarrow \Tmc \models \alpha'$), 
then we can finish the proof.

Recall the construction of $\Gmc_{\Tmc}$: For every node 
$v \in D$ of $\Tmc = (D,\to,\rho)$ we introduce $m := |\Var_\varphi|$ copies
$(v,x)$ for $x \in \Var_\varphi$.
The $\Smc$-relations between these nodes are determined
by the propositions $p_1, \ldots, p_n$: The interpretation of 
$r \in \Smc$ contains all $k$-tuples ($k = \arity{r}$)
$((su_1,y_1), \ldots, (su_k, y_k))$ for which
there exist $1 \leq i \leq n$ and $u \in [1,d]^*$ with $|u|=d_i$,
$p_i \in \rho(su)$, $R_i = r(\neX^{j_1}y_1,\ldots,\neX^{j_k}y_k)$,
and $u_t = u[:j_t]$ for $1 \leq t \leq k$.
This is a particular case of an
$\MSO$-transduction \cite{Courcelle91a} with copy number $m$. It is therefore
possible to compute from a given $\MSO$-sentence $\eta$ over the signature 
$\Smc$ an $\MSO$-sentence $\eta'$  
such that $\Gmc_{\Tmc} \models \eta \Leftrightarrow 
\Tmc \models \eta'$. But the problem is that in our situation
$\eta$ is the $\Bool(\MSO,\WMSOB)$-formula $\alpha$, and it is not
clear whether $\MSO$-transductions (or even first-order
interpretations) are compatible with the logic $\WMSOB$. 
Nevertheless, there is a simple solution.
Let $\Var_\varphi = \{x_1, \ldots, x_m\}$.
From a Kripke $d$-tree $\Tmc = ([1,d]^*,\to,\rho)$ we build an extended $(d+m)$-Kripke
tree $\Tmc^e = ([1,d+m]^*,\to,\rho^e)$ as follows: 
Let us fix new propositions $q_1,\ldots,q_m$ (one for each variable $x_i$) that do not
occur in the $\MSO$-sentence $\psi$ and such that $\rho(v) \cap \{q_1,\ldots,q_m\} =
\emptyset$ for all $v \in [1,d]^*$. 
We define the new labeling function $\rho^e$ as
follows: 
\begin{eqnarray*}
\rho^e(v) & = & \rho(v) \text{ for } v \in [1,d]^* \\
\rho^e(vi) & = & \{q_{i-d}\} \text{ for } v \in [1,d]^*, d+1 \leq i \leq d+m \\
\rho^e(viu) & = & \emptyset \text{ for } v \in [1,d]^*, d+1 \leq i
\leq d+m, u \in [1,d+m]^+
\end{eqnarray*}
It is easy to write down an MSO-sentence $\beta$ such that 
for every $(d+m)$-Kripke tree $\Tmc'$ we have
$\Tmc' \models \beta$ if and only if $\Tmc' \cong \Tmc^e$ 
for some Kripke $d$-tree $\Tmc$.
Moreover, since the old Kripke $d$-tree $\Tmc$ is $\MSO$-definable within
$\Tmc^e$, we can construct from the $\MSO$-sentence $\psi$
a new $\MSO$-sentence $\psi^e$ such that 
$\Tmc\models \psi$ if and only if $\Tmc^e \models \psi^e$.
Finally, let $q(x) = \bigvee_{i=1}^m q_i(x)$. Then, the nodes of 
$\Gmc_{\Tmc}$ are in a natural bijection with the nodes of $\Tmc^e$ that
satisfy $q(x)$: If $\Tmc^e \models q(u)$ for $u \in [1,d+m]^*$, then 
there is a unique $i \in [1,m]$ such that 
$\Tmc^e \models q_i(u)$ and $u = v (i+d)$. Then we associate
the node $u$ with node $(v,x_i)$ of $\Gmc_{\Tmc}$.
By relativizing all quantifiers in  the
$\Bool(\MSO,\WMSOB)$-formula $\alpha$ to $q(x)$,
we can construct a
$\Bool(\MSO,\WMSOB)$-formula $\alpha^e$ such that
$\Gmc_{\Tmc} \models \alpha$ if and only if 
$\Tmc^e \models \alpha^e$.

It follows that there is
a Kripke $d$-tree $\Tmc$ such
that $\Tmc\models \psi$ and $\Gmc_{\Tmc} \models \alpha$ 
if and only if there is
a Kripke $(d+m)$-tree $\Tmc'$ such
that $\Tmc'\models (\beta \wedge \psi^e \wedge \alpha^e)$.
Since $\beta \wedge \psi^e \wedge \alpha^e$ is a 
$\Bool(\MSO,\WMSOB)$-formula, the latter is decidable by 
Thm.~\ref{thm-bojan-to2}.

\section{Concrete domains over the integers}

The main technical result of this section is:
\begin{proposition}\label{second}
 $\EZmc$ from \eqref{zmc^m} has the property
$\EHomDef(\Bool(\MSO,\WMSOB))$.
\end{proposition}
Since $\EZmc$ is negation-closed 
(see Ex.~\ref{ex1})
our main result follows by Thm.~\ref{thm-EHOMDef-SAT}:
\begin{theorem}\label{thm-main}
 $\SAT(\EZmc)$ is decidable.
\end{theorem}
We prove Prop.~\ref{second} in three steps. First, we show  that  the structure $(\mathbb{Z}, <)$ has the
property \EHomDef($\WMSOB$). Then we extend this result to the
structure $(\mathbb{Z}, <,=)$ and, finally, to the full structure $\EZmc$.

As a preparation of the proof,  we first define some terminology and
then we characterize structures that allow homomorphisms to $(\mathbb{Z}, <)$ in
terms of their paths.  Let $\Amc = (A,I)$ be a countable  $\{<\}$-structure. 
We identify $\Amc$ with the directed graph $(A,E)$ where $E = I(<)$.
When talking about paths, we always refer to finite directed $E$-paths.
The length of a path $(a_0, a_1, \ldots, a_n)$ (i.e., $(a_{i-1}, a_{i}) \in E$ for 
$1 \leq i \leq n$) is $n$. 
For $S \subseteq A$ and  $x \in A \setminus S$, a path
 from $x$ to $S$ is a path from $x$ to some node $y \in S$.
 A path from $S$ to $x$ is defined in a symmetric way.

\begin{lemma}\label{lemma1}
We have $\Amc  \homom (\mathbb{Z}, <)$ if and only if 
\begin{enumerate}[(H1)]
\item \label{H1} $\Amc$ does not contain cycles, and
\item \label{H2} for all $a,b \in A$ there is $c \in \mathbb{N}$
  such that the  length of all paths from $a$ to $b$ is bounded by $c$. 
\end{enumerate}
\end{lemma}

\begin{proof}
Let us first show the ``only if'' direction of the lemma.
Suppose $h$ is a homomorphism from $\Amc$ to $(\mathbb{Z}, <)$. 
The presence of a cycle $(a_0 ,\ldots a_{k-1})$ in $\Amc$ ($k \geq 1$, $(a_i, a_{i+1 \text{ mod } k}) \in E$ for $0 \leq i \leq k-1$)
would imply the existence of integers $z_0 ,\ldots z_{k-1}$ with $z_i < z_{i+1 \text{ mod } k}$ for $0 \leq i \leq k-1$
(where $z_i = h(a_i)$), which is not possible. Hence, (H1) holds.

Suppose now that $a, b \in A$ are such that for every $n$ there is a path of length at least $n$ 
from $a$ to $b$. If $d = h(b) - h(a)$, we can find a path $(a_0, a_1 \ldots, a_k)$ with
$a_0=a$, $a_k = b$ and $k > d$. Since $h$ is a homomorphism, this path
will be mapped to an increasing sequence of integers $h(a)=h(a_0) < h(a_1)
< \cdots < h(a_k)=h(b)$. But this contradicts $h(b) - h(a) = d < k$. Hence, (H2) holds.

For the ``if'' direction of the lemma assume that $\Amc$ is acyclic (property (H1)) and that
(H2) holds. Fix an enumeration  $a_0, a_1, a_2,
\ldots$ of the countable set $A$. For $n \geq 0$ let
$S_n \defeq\{ a\in A\mid \exists i,j \leq n : (a_i,a), (a,a_j) \in E^* \}$, which
has the following properties:
\begin{enumerate}[(P1)]
\item \label{P1}  $S_n$ is convex  w.r.t.~the partial order $E^*$: 
If $a,c \in S_n$  and $(a,b), (b,c) \in E^*$, then $b \in S_n$.
\item \label{P2}  For $a\in A\setminus S_n$ all paths between $a$ and $S_n$ are 
  ``one-way'', i.e.,  there do not exist $b,c \in S_n$ such that $(b,a), (a,c) \in E^*$.
  This follows from (P\ref{P1}).
\item \label{P4} For all $a\in A \setminus S_n$ there exists a bound
  $c \in\mathbb{N}$ such that all paths between $a$ and $S_n$ have length at most $c$.  Let $c_n^a \in \mathbb{N}$
  be the smallest such bound (hence, we have $c_n^a=0$ if there do not exist paths between $a$ and $S_n$). 
\end{enumerate}
To see (P\ref{P4}), assume that there only exist paths from $S_n$ to $a$ but not the other way
round (see  (P\ref{P2})); the other case is symmetric. If there is no bound on the length of paths from
$S_n$ to $a$, then by definition of $S_n$, there is no bound on the length of paths from
$\{a_0,\ldots,a_n\}$ to $a$. By the pigeon principle, there exists $0 \leq i \leq n$ such that 
there is no bound on the length of paths from
$a_i$ to $a$.  But this contradicts property (H2).

We build our homomorphism $h$ inductively. For every $n \geq 0$ 
we define functions $h_n:S_n \to \mathbb{Z}$ such that
the following invariants hold for all $n \geq 0$.
\begin{enumerate}[(\text{I}1)]
\item\label{P0}  If $n > 0$ then
$h_{n}(a)=h_{n-1}(a)$ for all $a \in S_{n-1}$
\item \label{P3} $h_n(S_n)$ is bounded in $\mathbb{Z}$, i.e., there exist $z_1, z_2 \in \mathbb{Z}$
such that $h_n(S_n) \subseteq [z_1,z_2]$.
\item \label{P5} $h_n$ is a homomorphism from the subgraph $(S_n,  E \cap (S_n \times S_n))$ to $(\mathbb{Z}, <)$.
\end{enumerate}
For $n=0$ we have $S_0 = \{a_0\}$. We set $h_0(a_0) = 0$ (any other
integer  would 
be also fine). Properties (I\ref{P0})--(I\ref{P5}) are easily verified.
For $n > 0$, there are four cases. 

\smallskip
\noindent
{\em Case 1.} $a_n\in \SNN$, thus $S_n=\SNN$.
We set $h_n = h_{n-1}$. Clearly, (I\ref{P0})--(I\ref{P5}) hold for $n$.

\smallskip
\noindent
{\em Case 2.} 
$a_n\notin \SNN$ and there is no path from $a_n$ to $\SNN$ or
vice versa. We set 
$h_n(a_n) \defeq 0 $ (and $S_n=\SNN \cup\{a_n\}$).
In this case (I\ref{P0})--(I\ref{P5}) follow easily from the
induction hypothesis. 

\smallskip
\noindent
{\em Case 3.}  
$a_n\notin \SNN$ and 
there exist paths from $a_n$ to $\SNN$. Then, by 
(P\ref{P2}) there do not exist paths from $\SNN$ to $a_n$.
Hence, we have 
$$
S_n = \SNN \cup \{ a \in A \mid \exists b \in \SNN : (a_n,a), (a,b) \in E^* \}. 
$$
We have to assign a value $h_n(a)$ for all 
$a \in A \setminus \SNN$ that lie along a path from $a_n$ to $\SNN$.
By (I\ref{P3}) there exist $z_1, z_2  \in \mathbb{Z}$ with $h_{n-1}(\SNN) \subseteq [z_1,z_2]$. 
Recall the definition of $c_{n-1}^a$ from (P\ref{P4}). 
For all $a \in A \setminus \SNN$ that lie on a path from $a_n$ to $\SNN$, 
we set $h_n(a) := z_1 - c_{n-1}^a$.  Since there are paths from $a$ to $\SNN$,
we have $c_{n-1}^a > 0$. Hence, for all $a \in
S_n \setminus \SNN$, $h_n(a) < z_1$. 
Let us check that $h_n: S_n \to \mathbb{Z}$ satisfy (I\ref{P0})-- (I\ref{P5}):
Invariant (I\ref{P0}) holds by definition of $h_n$. For (I\ref{P3}) note that
$h_n(S_n) \subseteq [z_1 - c_{n-1}^{a_n},z_2]$.

 It remains to show  (I\ref{P5}), i.e.,  that $h_n$ is a homomorphism 
 from $(S_n,  E \cap (S_n \times S_n))$ to $(\mathbb{Z}, <)$.
 Hence, we have to show that  $h(b_1) < h(b_2)$ 
 for all  $(b_1,b_2) \in E \cap (S_n \times S_n)$.
  \begin{itemize}
  \item If $b_1,b_2 \in \SNN$, then  
    $h_n(b_1)= h_{n-1}(b_1)  <  h_{n-1}(b_2) = h_n(b_2)$ by induction
    hypothesis. 
  \item 
    If $b_1 \in S_n\setminus \SNN$ and $b_2 \in \SNN$, we know that 
    $h_n(b_2)=h_{n-1}(b_2) \geq z_1$ while $h_n(b_1)<z_1$
    by construction. This directly implies $h_n(b_1) < h_n(b_2)$. 
  \item 
    If $b_2 \in S_n\setminus \SNN$ and $b_1 \in \SNN$, then 
    $(b_1,b_2)\in E$ and by  assumption $b_2$ must be on 
    a path from $a_n$ to $\SNN$ which contradicts (P\ref{P2}). 
  \item 
    If both $b_1$ and $b_2$ belong to $S_n \setminus \SNN$
    then $h_n(b_i) \defeq z_1 - c_{n-1}^{b_i}$
    for $i\in \{1,2\}$  
    Since $(b_1,b_2) \in E$, we have $c_{n-1}^{b_1} > c_{n-1}^{b_2}$. 
    This implies $h_n(b_1) < h_n (b_2)$. 
    \end{itemize}
{\em Case 4.}  
$a_n\notin \SNN$ and 
there exist paths from $\SNN$ to $a_n$. 
For all $a \in S_n \setminus \SNN =  \{  a \in A \setminus \SNN \mid a \text{ belongs to a
  path from } \SNN \text{ to } a_n \}$, set
  $h_n(a) = z_2 + c_{n-1}^a$. The rest of the argument goes analogously to Case~3.
  
  This concludes the construction of  $h_n$.
 By (I\ref{P0}) limit function 
 $h = \bigcup_{i\in \mathbb{N}} h_i$ exists. By (I\ref{P5}) and $A
 =\bigcup_{i \in \mathbb{N}} S_i$, 
 $h$ is a homomorphism from $\Amc$ to $(\mathbb{Z}, <)$.
\qed
\end{proof}

\begin{proposition}\label{first}
 $(\mathbb{Z}, <)$ has the property
\EHomDef$(\WMSOB)$.
\end{proposition}
\begin{proof}
We  translate the 
conditions (H1) and (H2) from Lemma~\ref{lemma1} into  $\WMSOB$.
Cycles are excluded
by the sentence $\neg \ExistsCycle_<$ (Example~\ref{exa:WMSOBFormulas}). Moreover,   
for an acyclic $\{<\}$-structure $\Amc$ we have
$\Amc \models \forall x \forall y \; \BoundedPaths_<(x,y)$ (see also Example~\ref{exa:WMSOBFormulas}) if and only if 
for all $a,b\in A$ there is a bound $b \in \mathbb{N}$ 
on the length of  paths from $a$ to $b$. 
Thus, $\Amc \homom (\mathbb{Z}, <)$ if and only if
$\Amc \models \neg \ExistsCycle_< \land  \forall x \forall y\  \BoundedPaths_<(x,y)$.
 \qed
\end{proof}
Next, we extend Prop.~\ref{first} to the negation-closed structure $(\mathbb{Z},<,=)$.  
To do so let us fix a countable $\{<,=\}$-structure 
$\Amc = (A,I)$. Note that $I(=)$ is not 
necessarily the identity relation on $A$. Let $\sim \; = (I(=) \cup I(=)^{-1})^*$ be the smallest
equivalence relation on $A$ that contains $I(=)$.
Since $\sim$ is the reflexive and transitive closure of the first-order
definable relation $I(=) \cup I(=)^{-1}$, we can construct a $\WMSO$-formula 
$\tilde{\varphi}(x,y)$  (using the $\reach$-construction from Ex.~\ref{exa:WMSOBFormulas})
that defines $\sim$.
Let 
\begin{align} \label{def-E_<}
&E_{<} = \; \sim \circ \, I(<) \, \circ \sim \text{i.e., the relation defined
by the formula }\\
&\label{formula-phi_<}
\varphi_<(x,y) = \exists u\; \exists v \; (\tilde{\varphi}(x,u) \land u<v \land \tilde{\varphi}(v,y)).
\end{align}
With $\tilde{\Amc} = (\tilde{A}, \tilde{I})$ we denote the $\sim$-quotient of $\Amc$: It is a $\{<\}$-structure,
its domain is the set $\tilde{A} = \{ [a]_\sim \mid a \in A\}$  of all $\sim$-equivalence classes.
and for two equivalence classes $[a]_\sim$ and $[b]_\sim$ we have 
$([a]_\sim, [b]_\sim) \in \tilde{I}(<)$ iff there are $a' \sim a$ and $b'\sim b$ such that 
$(a',b') \in I(<)$. Let us write $[a]$ for $[a]_\sim$.
We have:

\begin{lemma} \label{lemma-quotient}
$\Amc \homom (\mathbb{Z},<,=)$ if and only
$\tilde{\Amc} \homom (\mathbb{Z},<)$. 
\end{lemma}
\begin{proof}
Suppose $h:\Amc \rightarrow (\mathbb{Z},<,=)$ is a homomorphism. 
Since $a \sim b$ implies $h(a)=h(b)$, we can define a mapping $h' : \tilde{A}
\to \mathbb{Z}$ by  $h'([a])=h(a)$ for all $[a] \in \tilde{A}$. 
Now let $a,b\in A$ such that 
$([a],[b]) \in \tilde{I}(<)$. Then there are $a' \sim a$ and $b'\sim b$
such that $(a',b') \in I(<)$.
Therefore $h'([a])=h(a') < h(b')=
h'([b])$. Hence $h'$ is a homomorphism.
 
For the other direction, suppose that 
$h: \tilde{\Amc} \to (\mathbb{Z},<)$ is a homomorphism.
We define $h': A \rightarrow \mathbb{Z}$ by
$h'(a) = h([a])$ for all $a \in A$. If $a,b \in A$ 
are such that $(a,b) \in I(=)$ then $[a]=[b]$ and therefore
$h'(a)=h'(b)$. If $a,b \in A$ are such that $(a,b) \in I(<)$ then
$([a],[b]) \in \tilde{I}(<)$, whence $h'(a) = h([a]) < h([b])=h'(b)$.
Thus,  $h'$ is a homomorphism.
\qed
\end{proof}

In the next lemma, we translate the conditions for the 
existence of a homomorphism from $\tilde{\Amc}$
to $(\mathbb{Z},<)$ into conditions in terms of
$\Amc$.

\begin{lemma}  \label{lemma-quotient-2}
The following conditions are
equivalent: 
\begin{itemize}
\item $\tilde{\Amc}$ satisfies the conditions (H1) and (H2) from
  Lemma~\ref{lemma1}.
 \item The graph $(A,E_<)$ is acyclic  and for all $a,b \in A$ there
   is a bound $c \in \mathbb{N}$ such that  all $E_<$-paths 
   from $a$ to $b$ have length at most $c$.
 \end{itemize}
\end{lemma}
\begin{proof}
  The proof is straightforward once we notice that any path in
$\tilde{\Amc}$ corresponds to a path in $(A,E_<)$. 
More precisely,  $([a_0], \ldots, [a_k])$ is a path in
the graph $\tilde{\Amc}$ (i.e.,   $([a_i],[a_{i+1}]) \in \tilde{I}(<)$  for all 
$0\leq i < k$) if and only if $(a_0, \ldots, a_k)$ is a path in $(A,E_<)$.
It follows directly, that there is a cycle in $(A,E_<)$
if and only if there is a cycle in $\tilde{\Amc}$.
Moreover,  for all  $a,b \in A$, there is a bound $c \in \mathbb{N}$ 
on the length of $E_<$-paths from $a$ to $b$ if and only if 
there is a bound on the length of paths  between $[a]$ and
$[b]$ in $\tilde{\Amc}$.
\qed
\end{proof}

\begin{proposition}\label{equal}
$(\mathbb{Z}, <, =)$ has the property \EHomDef$(\WMSOB)$.
\end{proposition}

\begin{proof}
Our aim is to find a $(\WMSOB)$-formula $\phi$ such
that for all  $\{<,=\}$-structures $\Amc$, $\Amc \models \phi$ if and only if $\Amc \homom (\mathbb{Z}, <, =)$.
Let $\Amc = (A,I)$ be a $\{<,=\}$-structure.  We use the notations 
introduced before Lemma~\ref{lemma-quotient}.
By Lemma~\ref{lemma-quotient} and \ref{lemma-quotient-2} we have to construct
a $(\WMSOB)$-formula expressing that $\Amc$ has no $E_<$-cycles and for all $a,b \in A$ there
is a bound $c \in \mathbb{N}$ on the length of  $E_<$-paths 
from $a$ to $b$. For this, we can use the formula constructed in the proof 
of Prop.~\ref{first} with $<$ replaced by the formula $\varphi_<$ from \eqref{formula-phi_<}.
\qed
\end{proof}
We will later also need the following variants of Prop.~\ref{equal}:

\begin{proposition}\label{equal-N}
$(\mathbb{N}, <, =)$ and $(\mathbb{Z}\setminus\mathbb{N}, <, =)$
have  property \EHomDef$(\WMSOB)$.
\end{proposition}

\begin{proof}
We prove the proposition only for $(\mathbb{N}, <, =)$, the statement for $(\mathbb{Z}\setminus\mathbb{N}, <, =)$
can be shown analogously. Let $\Amc=(A,I)$ be a $\{<,=\}$-structure. 
Define the relation $E_<$ as in \eqref{def-E_<}.
By adapting our proof for Prop.~\ref{equal}, one can show that 
$\Amc \homom (\mathbb{N}, <, =)$
if and only if $\Amc$ does not contain
$E_<$-cycles and for each $a\in A$ there is a bound $c$ such that any
$E_<$-path from some node of $A$ to $a$ has length at most $c$.
This is (\WMSOB)-expressible by the sentence
$\neg \ExistsCycle_{\varphi_<} \land  \forall y  \; \Bound Z \; \exists x \; \Path_{\varphi_<}(x,y,Z)$.
\qed
\end{proof}
In the rest of this section, we prove Prop.~\ref{equal} for the full structure $\EZmc$ from \eqref{zmc^m}, which
is defined over the  infinite signature $\Smc = \{ <, = \} \cup \{ =_c \mid c \in \mathbb{Z}\}
\cup \{ \equiv_{a,b} \mid 0 \leq a < b\}$.
By the definition of 
\EHomDef$(\Bool(\MSO,\WMSOB))$ we have to compute from a finite subsignature 
$\sigma \subseteq \Smc$ a $\Bool(\MSO,\WMSOB)$-sentence $\varphi_\sigma$ that
defines the existence of a homomorphism to $\EZmc$ when interpreted over a $\sigma$-structure 
$\Amc$. Hence, let us fix a finite subsignature $\sigma \subseteq
\Smc$. We can assume that
%$\sigma  = \{ <, = \}  \cup \{ =_c \mid c \in C\}
%\cup \{ \equiv_{a,b} \mid b \in D, 0 \leq a < b\}$
\begin{equation*} 
\sigma  = \{ <, = \}  \cup \{ =_c \mid c \in C\}
\cup \{ \equiv_{a,b} \mid b \in D, 0 \leq a < b\}
\end{equation*}
for finite non-empty sets $C \subseteq \mathbb{Z}$ and $D \subseteq \mathbb{N}\setminus\{0,1\}$. 
Define $m=\min(C)$ and $M=\max(C)$.
W.l.o.g. we can assume that $m \leq 0$ and $M \geq 0$.
Let $\Amc =(A,I)$ be a countable $\sigma$-structure.
In order to not confuse the relation $I(=)$ with the identity relation on $A$, we 
write in the following $E_=(x,y)$ for the atomic formula expressing that $(x,y)$ 
belongs to the relation $I(=)$. Similarly, we write $E_c(x)$ for the atomic formula expressing that $x \in I(=_c)$. Instead of $\equiv_{a,b}\!\!(x)$ we write $x \equiv a$ mod $b$.

Define $x \leq y \Leftrightarrow (x<y \lor E_=(x,y) \lor E_=(y,x))$ and the  \MSO-formula 
$$
  \varphi_{\text{bounded}}(x) =  \exists y \; \exists z \big( \bigvee_{c \in C} E_c(y)
  \wedge \bigvee_{c \in C} E_c(z)  \wedge \reach_{\leq}(y,x)
  \wedge \reach_{\leq}(x,z) \big).
 $$
Let $B=\{a\in A\mid \Amc\models \varphi_{\text{bounded}}(a)\}$. We call
the induced substructure $\Bmc := \Amc{\restriction}_B$ the 
``bounded'' part of $\Amc$. Every homomorphism from $\Bmc$ to $\EZmc$ 
has to map $B$ to the interval $[m,M]$. Thus, a homomorphism $h :  \Bmc \to \EZmc$ can be identified 
  with a partition of $B$ into $M-m+1$ sets $B_m, \ldots, B_M$, where 
  $B_i = \{ a \in B \mid h(a)=i \}$.  It follows that:

\begin{lemma} \label{lem:Homomorphism-BoundedPart}
  There is an \MSO-sentence $\varphi_B$ such that
  for every $\Smc$-structure $\Amc$ with bounded part $\Bmc$, we have $\Bmc \homom \EZmc$ 
  if and only if $\Amc\models\varphi_B$.
\end{lemma}
\begin{proof}
    By definition of the bounded part, 
  any homomorphism from $\Bmc$ to $\EZmc$ maps all elements of $B$ to a value
  from the interval $[m,M]$. Thus, a homomorphism $h :  \Bmc \to \EZmc$ can be identified 
  with a partition of $B$ into $M-m+1$ sets $B_m, \ldots, B_M$, where 
  $B_i = \{ a \in B \mid h(a)=i \}$. Hence, the \MSO-sentence states that
  there exists a partition of $B$ into  $M-m+1$ sets $B_m, \ldots, B_M$ such 
  that the corresponding mapping $h : B \to [m,M]$ preserves all relations from
  $\sigma$. For this we define formulas that express the following, where $\overline{X} = (X_m, \ldots, X_M)$
  is a tuple of $M-m+1$ many second-order variables.
\begin{itemize}
  \item $\varphi_{\text{part}}(\overline{X})$ expresses that $\overline{X}$ forms  a  finite partition.
  \item $\varphi_<(\overline{X})$ expresses that the partition
    preserves the relation $I(<)$.
  \item $\varphi_=(\overline{X})$ expresses that the partition
    preserves the relation $I(=)$.
  \item $\varphi_{\text{const}}(\overline{X})$ expresses that the partition
    preserves all relations $I(=_c)$.
  \item $\varphi_{\text{mod}}(\overline{X})$ expresses that the partition
    preserves all relations $I(\equiv_{a,b})$.
  \end{itemize}
 These formulas can be defined as follows:
  \begin{eqnarray*}
    \varphi_{\text{part}} & = &   \forall x 
     \bigvee_{i \in [m,M]}  \Big(  x \in X_i   \wedge  
      \bigwedge_{\substack{j \in [m,M] \\ i \neq j}} x \not\in   X_j
    \Big),\\
    \varphi_{<}   & \ = \ &  \forall x \;\forall y \bigwedge_{\substack{i,j
        \in [m,M] \\ i \geq j}}  
    \neg (x<y \wedge x \in X_i \wedge y \in X_j),\\ 
    \varphi_{=}   & = &  \forall x \; \forall y 
    \bigwedge_{\substack{i,j \in [m,M] \\ i \neq j}} \neg (E_=(x,y) \wedge x
    \in X_i \wedge y \in X_j), \\ 
    \varphi_{\text{const}} & = &    \forall x 
    \bigwedge_{c\in C} \big( E_c(x) \to x \in X_c \big),  \\  
    \varphi_{\text{mod}} & = &  \forall x
    \bigwedge_{0 \leq a < b \in D} \Big(
    x \equiv a \text{ mod } b  \rightarrow  
    \bigvee_{\substack{i \in [m,M] \\  i \equiv a \text{ mod } b}}  x \in X_i \Big).    
  \end{eqnarray*} 
  Let $\psi = \exists X_m  \cdots \exists X_M
    (\varphi_{\text{part}} \land \varphi_{<} 
    \land \varphi_{=} \land \varphi_{\text{const}}
    \land \varphi_{\text{mod}})$ and let $\varphi_B$ be the relativization
  $\psi$ to the bounded part defined by $\varphi_{\text{bounded}}(x)$. 
  Then, $\Amc \models \varphi_B$ if and only if $\Bmc \models \psi$
  if and only if there is a homomorphism $h_B : \Bmc \to \EZmc$.
  \qed
\end{proof}
Similar to $B$ we define three other parts of a
$\sigma$-structure by the \WMSO-formulas
\begin{eqnarray*}
  \varphi_{\text{greater}}(x) &=& \neg \varphi_{\text{bounded}}(x) \wedge
   \exists y \; \big(  \varphi_{\text{bounded}}(y) \wedge \reach_{\leq}(y,x) \big), \\
 \varphi_{\text{smaller}}(x) &=& \neg \varphi_{\text{bounded}}(x) \wedge
   \exists y \; \big(  \varphi_{\text{bounded}}(y) \wedge \reach_{\leq}(x,y) \big), \\  
  \varphi_{\text{rest}}(x) &=&  \neg (\varphi_{\text{bounded}}(x) \vee \varphi_{\text{greater}}(x)  \vee 
  \varphi_{\text{smaller}}(x)) .
\end{eqnarray*}
Moreover, let $G=\{a\in A\mid\Amc\models\varphi_{\text{greater}}(a)\}$,
$S=\{a\in A\mid\Amc\models\varphi_{\text{smaller}}(a)\}$, and 
$R=\{a\in A\mid\Amc\models\varphi_{\text{rest}}(a)\}$.
Let $\ENmc = \EZmc{\restriction}_{\mathbb{N}}$ and $\ENegmc =
\EZmc{\restriction}_{\mathbb{Z}\setminus\mathbb{N}}$. 
Then we have:

\begin{lemma} \label{lem:splitinFour}
$\Amc\homom \EZmc$ iff
  $\left(\Bmc \homom \EZmc,
  \Amc{\restriction}_{G \cup S \cup R} \homom \EZmc,  
  \Amc{\restriction}_G \homom \ENmc,\text{ and }
  \Amc{\restriction}_S \homom \ENegmc\right)$.
\end{lemma}
\begin{proof}
  The ``only if'' direction is straightforward. Just note that for a
  homomorphism $h:\Amc\to \EZmc$, $h(G)$ is bounded below by $m$ and 
  $h(S)$ is bounded above by $M$.
  
  For the ``if'' direction, assume that there are 
  \begin{itemize}
  \item a homomorphism $h_B: \Bmc \to \EZmc$,
  \item a homomorphism $h_R : \Amc{\restriction}_{G \cup S \cup R} \to \EZmc$,  
  \item a homomorphism $h_G : \Amc{\restriction}_G \to \ENmc$, and
  \item a homomorphism $h_S : \Amc{\restriction}_S \to \ENegmc$.
  \end{itemize}
  Let 
  $\delta = \prod_{b \in D} b \geq 1$
  and
  define $h:\Amc\to \EZmc$ by 
  \begin{align*}
    h(a)=
    \begin{cases}
      h_B(a) &\text{if }a\in B,\\
      h_R(a) &\text{if }a\in R, \\
      \max(h_R(a),h_G(a)) + \delta \cdot (M+1) & \text{if }a\in G, \\
      \min(h_R(a),h_S(a))  + \delta \cdot (m-1)  & \text{if }a\in S.
    \end{cases}
  \end{align*}
  Note that $M < h(a)$ for every $a \in G$ (recall that we assume $M
  \geq 0$) and thus
  \begin{equation} \label{eq-greater-B}
   \forall a \in B\; \forall a' \in G : h(a) < h(a').
  \end{equation}  
  Similarly, we have
  \begin{equation} \label{eq-smaller-B}
   \forall a \in S\; \forall a' \in B : h(a) < h(a').
  \end{equation} 
  Clearly, $(a,a') \in I(=)$ implies that $a$ and $a'$ belong to the same
  part ($B$, $G$, $S$, or $R$), which implies $h(a) = h(a')$. Moreover, if $(a,a') \in I(<)$, then we we have one of the
  following cases:
  \begin{enumerate}[(a)]
  \item $a,a'$ belong to the same part,
  \item $a\in S, a'\in G$,
  \item $a\in B, a'\in G$,
  \item $a\in S, a'\in B$,
  \item $a\in S, a'\in R$, 
  \item $a\in R, a'\in G$. 
  \end{enumerate}
  In cases (a), (b), (e), and (f) we get $h(a) < h(a')$
  by using the homomorphisms $h_B$, $h_G$, $h_S$, $h_R$.
  In cases (c) (resp., (d)) we get $h(a) < h(a')$ from
  \eqref{eq-greater-B}
  (resp., \eqref{eq-smaller-B}).
 Finally,
  the unary constant predicates and modulo predicates are preserved because we build the
  homomorphism from homomorphism that preserve these predicates. 
\qed
\end{proof}
We need some conventions on modulo constraints.
A sequence $(a_1,b_1), \ldots, (a_k,b_k)$ with 
$0 \leq a_i < b_i \in D$ for $1 \leq i \leq k$ is {\em  contradictory},
if there is no number $n\in \mathbb{N}$ such that
$n \equiv a_i$ mod $b_i$ for all $1 \leq i \leq k$.
In the following let $\mathsf{CS}_k$ denote the set of contradictory
sequences of length $k$. 
It is straightforward to show
that every contradictory sequence contains a contradictory subsequence of
length at most $\ell := \max\{2,|D|\}$. 

Recall that $\sim$ is the smallest 
equivalence relation containing $I(=)$  and that $\sim$ is defined
by the \WMSO-formula $\tilde{\varphi}(x,y)$.
We call a $\sigma$-structure $\Amc=(A,I)$ \emph{modulo contradicting}
if there is a $\sim$-class $[c]$, elements $c_1, c_2, \ldots, c_k\in [c]$,
and a contradictory sequence 
$(a_1,b_1),  \ldots, (a_k,b_k)$ such that
$c_i \in I(\equiv_{a_i,b_i})$ for all $1 \leq i \leq k$.

The following \WMSO-formula $\varphi_{\text{modcon}}$ expresses
that a $\sigma$-structure is modulo contradicting, where 
we write $s_a(j)$ (resp.~$s_b(j)$) 
for the first (resp.~second) entry of the $j$-th element of
the sequence $s \in \mathsf{CS}_k$:

\begin{align*}
  \varphi_{\text{modcon}} = \bigvee_{2\leq k \leq \ell} \bigvee_{s\in
    \mathsf{CS}_k} \exists x_1 \cdots \exists x_k  
  \bigwedge_{i,j\leq k} \tilde{\varphi}(x_i,x_j) \land  \bigwedge_{j\leq k} x_j\equiv s_a(j) \text{ mod }  s_b(j)
\end{align*}

\begin{lemma} \label{lem:reducemodulos}
  Let $\sigma' = \sigma\setminus\{{=_c}\mid c\in\mathbb{Z}\}$.
  Let $\Amc=(A, I)$ be a $\sigma'$-structure. 
  \begin{itemize}
  \item $\Amc \homom \EZmc$ iff
    $\Amc$ is not modulo contradicting and
    $(A, I(<), I(=)) \homom (\mathbb{Z}, <, =)$.
  \item $\Amc \homom \ENmc$ iff 
    $\Amc$ is  not modulo contradicting and
   $(A, I(<), I(=)) \homom (\mathbb{N}, <, =)$.
  \end{itemize}
\end{lemma}
\begin{proof}
  The ``only if'' directions are obvious.
  For the ``if'' directions, assume that $g: (A, I(<), I(=)) \to
  (\mathbb{Z}, <,=)$ is a 
  homomorphism and that $\Amc$ is not  
  modulo contradicting. Let
  $$\delta = \prod_{b \in D} b.$$
  Hence, for each $c\in A$ there is
  a number $0 \leq m_c\leq \delta-1$ such that for all $d \sim c$, if 
  $d \in I(\equiv_{a,b})$ (where $0 \leq a < b \in D$) then $m_c
  \equiv a$ mod $b$. 
  Setting $h(c) = \delta \cdot g(c) + m_c$ we obtain a homomorphism
  $h:\Amc \to \EZmc$. The statement for $\ENmc$ follows in the same way.
  \qed
\end{proof}

{\em Proof of Prop.~\ref{second}.}
  Let $\Amc=(A, I)$ be a $\sigma$-structure. We defined a 
  partition of $A$ into $B, G, S$, and $R$. 
  Since membership in each of these sets is ($\WMSOB$)-definable, we can
  relativize any ($\WMSOB$)-formula to any of these sets. 
  For instance, we write  $\varphi^G$ for the relativization of
  $\varphi$ to the substructure induced by $G$. 
  Let $\varphi_B$ be the $\MSO$-formula from
  Lemma \ref{lem:Homomorphism-BoundedPart}, and
  for $C \in \{ \mathbb{Z},  \mathbb{N},  \mathbb{Z} \setminus \mathbb{N} \}$
  let  $\varphi_C$
  be a formula that expresses $\Amc \preceq (C, <, =)$, see Prop.~\ref{equal} and~\ref{equal-N}.
  Then $\Amc \models (\varphi_B\land 
    {\varphi_{\mathbb{Z}}^{G\cup S \cup R}} \land
    \varphi_{\mathbb{N}}^{G} \land 
    \varphi_{\mathbb{Z} \setminus \mathbb{N}}^{S}    \land
    \neg\varphi_{\text{modcon}})$
  iff $\Amc \homom \EZmc$ due to Lemmas \ref{lem:splitinFour} and 
  \ref{lem:reducemodulos}. \qed

\section{Extensions, Applications, Open Problems}

A simple adaptation of our  proof for $\EZmc$ shows 
that $\Qmc=(\mathbb{Q}, <, =, (=_q)_{q\in\mathbb{Q}})$ has 
the property $\EHomDef(\Bool(\MSO,\WMSOB))$ as well:
$\Amc  = (A,I) \homom \Qmc$ iff (i) $(A,E_<)$ is acyclic, where $E_<$ is defined 
as in \eqref{def-E_<},  (ii) there does not exist $(a,b) \in E_<^+$
(the transitive closure of $E_<$)
with $a \in I(=_p)$, $b \in I(=_q)$ and $q \leq p$, and (iii) 
there do not exist $a \sim b$ with
$a \in I(=_p)$, $b \in I(=_q)$, and $q \neq p$.

Let us finally state a simple preservation theorem for $\Amc$-satisfiability for $\CTL^*(\Smc)$.
Assume that $\Amc$ and $\Bmc$ are structures over countable signatures $\Smc_{\Amc}$ and
$\Smc_{\Bmc}$, respectively, and let $B$ be the domain of $\Bmc$. We say that
$\Amc$ is {\em existentially interpretable} in $\Bmc$ if there exist
$n \geq 1$  and quantifier-free first-order formulas $\varphi(y_1, \ldots, y_l, x_1, \ldots, x_n)$ and
$$
\varphi_r(z_1, \ldots, z_{l_r}, x_{1,1}, \ldots, x_{1,n}, \ldots, x_{\arity{r},1}, \ldots,x_{\arity{r},n}) \text{ for } r \in \Smc_{\Amc}
$$ 
over the signature $\Smc_{\Bmc}$, where the mapping $r \mapsto \varphi_r$ has to be computable,
such 
that $\Amc$ is isomorphic to the structure $(\{ \overline{b} \in B^n  \mid \exists \overline{c} \in 
B^l\colon\Bmc \models \varphi(\overline{c},\overline{b}) \}, I)$ with
\begin{gather*}
I(r) =  \{ (\overline{b}_1, \ldots, \overline{b}_{\arity{r}}) \in B^{\arity{r}n} \mid \exists \overline{c} \in B^{l_r}\colon\Bmc \models 
 \varphi_r(\overline{c}, \overline{b}_1, \ldots, \overline{b}_{\arity{r}}) \} \text{ for } r \in \Smc_{\Amc} .
\end{gather*} 

\begin{proposition} \label{prop-red-sat}
  If $\SAT(\Bmc)$ is decidable and $\Amc$ is existentially
  interpretable in $\Bmc$, then $\SAT(\Amc)$ is decidable too.
\end{proposition}
\begin{proof}
  Let $\psi$ be a
  $\CTL^*(\Smc_{\Amc})$-formula. Let $\Var_\psi$  be the set
  of constraint variables that occur in $\psi$. 
  We use the notations introduced before Prop.~\ref{prop-red-sat}.
  Let us choose new variables $x_i$, $y_{x,j}$, and $z_{r,k}$ for all
  $1 \leq i \leq n$,
  $x \in \Var_\psi$, $1 \leq j \leq l$, $r \in \Smc_{\Amc}$, and $1 \leq
  k \leq l_r$.
  Define the $\CTL^*(\Smc_{\Bmc})$-formula 
  $$
  \theta = \psi' \wedge \mathsf{AG} \bigwedge_{x \in \Var_\psi}
  \varphi(y_{x,1}, \ldots, y_{x,l}, x_1, \ldots, x_n) 
  $$
  ($\mathsf{G}$ is the derived temporal operator for `globally''),
  where $\psi'$ is obtained from $\psi$ by replacing in $\psi$ 
  every constraint 
  $$r(\neX^{i_1} x_1, \ldots, \neX^{i_{\arity{r}}} x_{\arity{r}})$$ by
  the boolean formula 
  \begin{align*}
    \varphi_r(\neX^d z_{r,1}, \ldots, \neX^d z_{r,l_r},
    \neX^{i_1}x_{1,1}, \ldots, \neX^{i_1}x_{1,n}, \ldots,
    \neX^{i_{\arity{r}}} x_{\arity{r},1}, \ldots,
    \neX^{i_{\arity{r}}}x_{\arity{r},n}), 
  \end{align*}
  where $d=\max\{i_1, \ldots,i_{\arity{r}}\}$.
  Using arguments similar to those from the proof of
  Lemma~\ref{lemma-neg-closed}, one can show that
  $\psi$ is $\Amc$-satisfiable if and only
  if $\theta$ is $\Bmc$-satisfiable.
  \qed
\end{proof}
Examples of structures $\Amc$ that are existentially 
interpretable  in $(\mathbb{Z},<,=)$, and hence have a decidable $\SAT(\Amc)$-problem
are (i) $(\mathbb{Z}^n, <_{\text{lex}}, =)$ (for  $n\geq 1$), 
where $<_{\text{lex}}$ denotes the strict lexicographic order on $n$-tuples
of integers, and (ii) the structure $\text{Allen}_{\mathbb{Z}}$, which consists 
of all $\mathbb{Z}$-intervals together with Allen's  relations $b$
(before), $a$ (after), $m$ (meets), {\it mi} (met-by), $o$ (overlaps),
{\it oi} (overlapped by), $d$ (during), {\it di} (contains), $s$ (starts),
{\it si} (started by), $f$ (ends), {\it fi} (ended by). 
In artificial intelligence, Allen's relations are a popular tool for
representing temporal knowledge. 

Our technique can be also extended to the logic $\mathsf{ECTL}^*$
\cite{Thomas88,VardiW83}
that extends $\mathsf{CTL}^*$ by the ability to specify arbitrary \MSO-properties of infinite 
paths (instead of {\sf LTL}-properties for $\mathsf{CTL}^*$). 
For this one only has to extend Thm.~\ref{thm-tree-model} (tree model property
for $\mathsf{CTL}^*$ with constraints) to $\mathsf{ECTL}^*$ with constraints. The proof
is the same as in \cite{Gascon09}.

It remains open to determine the complexity of
$\mathsf{CTL}^*$-satisfiability 
with constraints over $\EZmc$, see the last paragraph in the 
introduction. Clearly, this problem is $2\mathsf{EXPTIME}$-hard due
to the known lower bound for $\mathsf{CTL}^*$-satisfiability.
To get an upper complexity bound, one should investigate the complexity of 
the emptiness problem for puzzles from \cite{BojanczykT12} (see Lemma~\ref{lem:PuzzleEmptyness}).
An interesting structure  for which the decidability status for
satisfiability of $\mathsf{CTL}^*$ with constraints 
is open, is  $(\{0,1\}^*, \leq_p, \not\leq_p)$, where
$\leq_p$ is the prefix order 
on words, and $\not\leq_p$ is its complement. It is not clear, whether
this structure has the  
property $\EHomDef(\Bool(\MSO,\WMSOB))$.

\paragraph{\bf Acknowledgments.} We are grateful to Szymon Toru\'nczyk
for fruitful discussions.

\end{document}